\documentclass{amsart}
\usepackage{amsmath}
\usepackage{amssymb,amsthm}
\usepackage{graphicx}
\numberwithin{equation}{section}
\newtheorem{thm}{Theorem}[section]
\newtheorem{rem}[thm]{Remark}

\newtheorem{lem}[thm]{Lemma}

\newtheorem{pro}[thm]{Proposition}

\title[Molecular predissociation resonances]{Molecular predissociation resonances below an energy level crossing}
\author{Sohei Ashida}
\begin{document}
\maketitle
\begin{abstract}
We study the resonances of $2\times 2$ systems of one dimensional Schr\"odinger operators which are related to the mathematical theory of molecular predissociation. We determine the precise positions of the resonances with real parts below the energy where bonding and anti-bonding potentials intersect transversally. In particular, we find that imaginary parts (widths) of the resonances are exponentially small and that the indices are determined by Agmon distances for the minimum of two potentials.
\end{abstract}
\section{Introduction}\label{zerothsec}
In this paper, we study precise positions of molecular predissociation resonances near the real axis as semiclassical asymptotics. More precisely, we consider resonances whose real parts are apart from a bottom of a well and a crossing of potentials. In particular, we obtain exponentialy small imaginary parts (widths) for these resonances.

In the theory of shape resonances of Schr\"odinger operators it is known that widths of the resonances have exponential bounds as $C_{\epsilon}e^{-2(1-\epsilon)S/h}$ where $\epsilon$ can be taken arbitrarily small, $h$ is the semiclassical parameter and $S$ is the Agmon distance between the bounded and unbounded regions where the potential is below the real part of the resonance (see, e.g., \cite{CDKS,HeS,HiS}). In one dimensional case, Servat \cite{Se} determine the precise positions of the resonances whose real parts are apart from the critical values of potentials using WKB constructions and considering the connection of the solutions and quantization conditions. In particular, it is proved that the exact order of the imaginary parts of the resonances is $he^{-2S/h}$.

When we consider the Schr\"odinger equation of molecules the study of the equation for electrons and nuclei is reduced to that of semiclassical system of Schr\"odinger-type operators by the Born--Oppenheimer approximation (see, e.g., \cite{KMSW,MM,MS}). In the Born--Oppenheimer approximation the semiclassical parameter $h$ is the square-root of the ratio of electronic to nuclear mass, and the potentials of the diagonal elements of the matrix of the Schr\"odinger operators describe electronic energy levels. At sufficiently low energies, since this system is scalar, numerous results from the semiclassical analysis of the Schr\"odinger operators can be applied. When several electronic levels are involved, states in different electronic levels interact due to the off-diagonal first order differential or pseudodifferential operators. 

In Martinez \cite{Ma}, Nakamura \cite{Na}, Baklouti \cite{Ba} and Grigis-Martinez \cite{GM2}, they study resonances for potentials that do not intersect and obtain exponential bound on their widths.
Klein \cite{Kl} studies the case of more than two intersecting potentials some of them forming wells to trap nuclei and the others being non-trapping. In this case, it is shown that the widths of the resonances with real parts converging to the bottom of the potential well have the exponential bound as in the case of usual Schr\"odinger operators with Agmon distance of the minimum of the two potentials. For intersecting two potentials, Grigis-Martinez \cite{GM} obtains the full asymptotic of the width of the resonance at the bottom of the well, showing that the exponential bound is optimal.

In Fujii\'e-Martinez-Watanabe \cite{FMW} they consider the resonances with real parts in the distance of order $h^{2/3}$ from a crossing of two potentials in one dimensional space. They construct the solution to the system on the left and right intervals from the crossing and consider the condition that solutions decaying on the left and those ontgoing on the right are connected at the crossing. The solutions to the system are constructed as series by successive approximation using the Yafaev's construction  for the solution of one dimensional Schr\"odinger equations (see \cite{Ya}) and showing that norms of operators including fundamental solutions are small for small $h$. Under an additional condition of ellipticity on the interaction they obtain the exact order $h^{5/3}$ of the widths of the resonances.

Here, as in \cite{FMW} we study $2\times 2$ matrix system, the diagonal part of which consists of semiclassical Schr\"odinger operators, and the off-diagonal parts of first order differential operators. We assume that two potentials $V_1$ and $V_2$ cross transversally and that below the crossing $V_1$ admits a well, while $V_2$ is non-trapping (see Fig. \ref{fig1}). We study the resonances with imaginary parts $\mathcal O(h)$ and with real parts apart from the crossing and the bottom of the well by a constant distance independent of $h$. The real parts of the resonances behave like eigenvalues obtained from usual Bohr--Sommerfeld quantization condition for the well, that is, the leading terms $e_k$ of the real parts of the resonances satisfy $\int_{a}^{b}\sqrt{e_k-V_1(x)}dx=(k+1/2)\pi h$, where $a$ and $b$ are the endpoints of the well. We also find that the widths of the resonances have exponential bounds as in \cite{Kl,GM} and under a condition on the interaction and the real part of the resonances they behave exactly like $h^2e^{-2S/h}$, where $S$ is the Agmon distance for $\min\{V_1,V_2\}$ between endpoints $b$ and $c$ of the region where the potentials are greater than $e_k$. In our model, the interaction is of the form $h(r_0(x)+hr_1(x)\partial_x)$ and the condition on the interaction is that $r_0(0)+r_1(0)\sqrt{V_1(0)+e_k}\neq 0$, where we assume the position of the crossing is $x=0$. This result should be compared with \cite[Theorem 2.6]{Se} for the case of one well. Paying attention to the shape of the graph of $\min\{V_1,V_2\}$ on the real axis our exponential bound is expected.

\begin{figure}
\includegraphics[width=10cm,clip,trim=0 150 0 200]{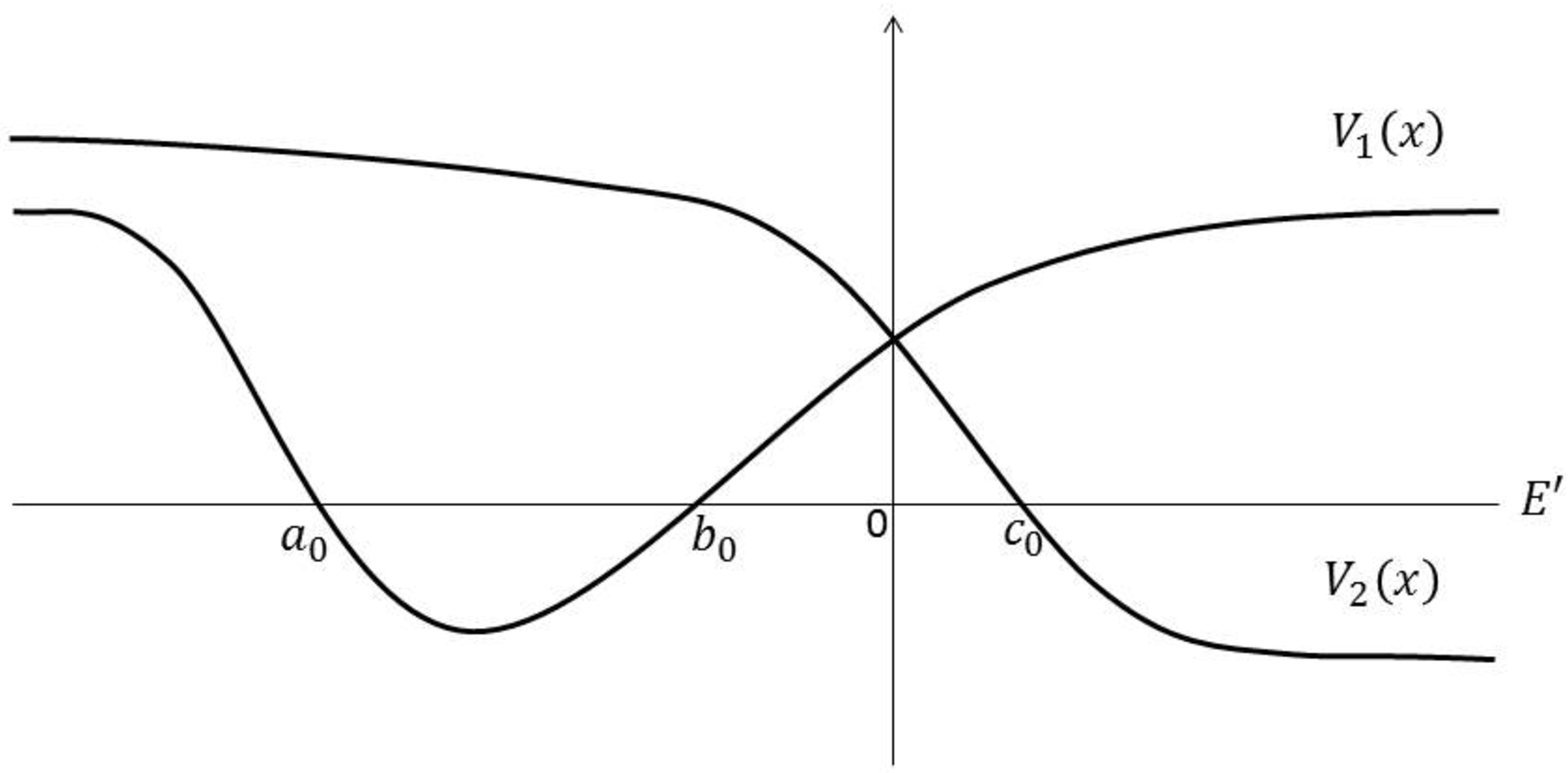}
\caption{The graph of the two potentials on the real axis}
\label{fig1}
\end{figure}

To prove our main theorem, we construct solutions to the system on four intervals, because if we formally construct solutions on two intervals as in \cite{FMW} the series in the solutions do not converge. We need to construct four solutions on each of the middle two intervals since the solutions on the left and right intervals are represented as linear combinations of the four solutions. To control the behavior of the series in the solutions, we need four fundamental solutions on each of the middle intervals. We choose the fundamental solutions so that formal leading terms in the solutions will really determine the order of the solution with respect to $h$. We study the connection of the solutions and calculate the transition matrix at the crossing and change the basis of the solutions in order that some of the elements of the transition matrix will be zero. This makes the calculation of the quantization condition simple and we can find exponentially small terms with the expected index.

The content of this paper is as follows. In Sect. \ref{firstsec} we give the assumptions and state our main result. In Sect. \ref{secondsec} we give some preliminaries. In particular, we introduce solutions to the scalar Schr\"odinger equations. In Sect. \ref{section2.1} we construct fundamental solutions on four intervals and give estimates on them. In Sect. \ref{thirdsec} we give the solutions to the system on the four intervals. In Sect. \ref{fourthsec} we study the connection of the solutions and change the bases of the solutions to make transition matrix suitable for the calculation of the quantization condition. In Sect. \ref{fifthsec} the quantization condition is given. Using the quantizaion condition in Sect. \ref{sixthsec} we prove the main theorem.

\section{Assumptions and results}\label{firstsec}
As in \cite{FMW}, we consider a $2\times2$ Schr\"odinger operator of the type,
\begin{equation}\label{myeq1.1}
Pu=Eu, P=\begin{pmatrix}
P_1& hW\\
hW^*& P_2
\end{pmatrix},
\end{equation}
where $P_j=-h^2\Delta+V_j(x), j=1,2$, $W=r_0(x)+hr_1(x)\partial_x$ and $W^*$ is the formal adjoint of $W$. 

We suppose the following conditions on $V_1(x)$, $V_2(x)$ (see Fig. \ref{fig1}) and $r_0(x)$, $r_1(x)$.\\
\textbf{Assumption (A1)} $V_1(x)$ and $V_2(x)$ are real-valued analytic functions on $\mathbb R$ and extend to holomorphic functions in the complex domain,
$$\Gamma=\{x\in \mathbb C;\lvert \mathrm{Im}\, x\rvert<\delta_0\langle \mathrm{Re}\, x\rangle\},$$
where $\delta_0>0$ is a constant and $\langle t\rangle:=(1+\lvert t\rvert^2)^{1/2}$.\\
\textbf{Assumption (A2)} For $j=1,2,$ $V_j$ admits limits as $\mathrm{Re}\, x\to \pm\infty$ in $\Gamma$ and there exists a real number $E'$ such that
\begin{align*}
&\lim_{\substack{
\mathrm{Re}\, x\to-\infty\\
x\in\Gamma}}
V_1(x)>E';\ \lim_{\substack{
\mathrm{Re}\, x\to-\infty\\
x\in\Gamma}}V_2(x)>E';\\
&\lim_{\substack{
\mathrm{Re}\, x\to+\infty\\
x\in\Gamma}}
V_1(x)>E';\ \lim_{\substack{
\mathrm{Re}\, x\to+\infty\\
x\in\Gamma}}V_2(x)<E'.
\end{align*}\\
\textbf{Assumption (A3)}
There exist real numbers $a_0<b_0<0<c_0$ such that
\begin{itemize}
\item $V_1>E'$ and $V_2>E'$ on $(-\infty,a_0)$;
\item $V_1<E'<V_2$ on $(a_0,b_0)$;
\item $E'<V_1<V_2$ on $(b_0,0)$;
\item $E'<V_2<V_1$ on $(0,c_0)$;
\item $V_2<E'<V_1$ on $(c_0,+\infty)$,
\end{itemize}
Moreover, one has
$$V_1'(a_0)<0,\ V_1'(b_0)>0,\ V_2'(c_0)<0,\ V_1'(0)>V_2'(0).$$
\textbf{Assumption (A4)}
$r_0$ and $r_1$ are bounded analytic functions on $\Gamma$, and $r_0(x)$ and $r_1(x)$ are real when $x$ is real.

As in \cite{FMW} we can define the resonances of $P$ as eigenvalues of the operator $P$ acting on $L^2(\mathbb R_{\theta})\oplus L^2(\mathbb R_{\theta})$ where $\mathbb R_{\theta}$ is a complex distortion of $\mathbb R$ that coincides with $e^{i\theta}\mathbb R$ for $x\gg1$. We denote by $\mathrm{Res}(P)$ the set of the resonances of $P$.

For $d>0$ small enough, we set $I:=[E'-d,E'+d]$. 
We fix $C_0$ arbitrarily large, and we study the resonances of $P$ lying in the set
$$\mathcal D_I:=\{E\in \mathbb C;\mathrm{Re}\, E\in I,\ -C_0h<\mathrm{Im}\, E<0\}.$$
For $E\in \mathcal D_I$, $V_1(x)=E$ has only two solutions and we denote the solution with the smaller real part and the larger one by $a=a(E)$ and $b=b(E)$ respectively. We also denote the unique solution of $V_2(x)=E$ for $E\in \mathcal D_I$ by $c=c(E)$. For $E\in \mathcal D_I$ we define the action,
$$\mathcal A(E):=\int_{a(E)}^{b(E)}\sqrt{E-V_1(t)}dt.$$
Then our result is the following.
\begin{thm}\label{main}
Under Assumptions (A1)-(A4), for $h>0$ small enough, one has,
$$\mathrm{Res}(P)\cap \mathcal D_I=\{E_k(h);k\in\mathbb Z\}\cap \mathcal D_I$$
where $E_k(h)$'s are complex numbers that satisfy,
$$\mathrm{Re}\, E_k(h)=e_k(h)+\mathcal O(h^{3/2})$$
\begin{align*}
\mathrm{Im}\, E_k(h)=&-\frac{h^2\pi}{4}\mathcal A'(e_k(h))^{-1}e^{-2S(e_k)/h}(V_1(0)-e_k(h))^{-1/2}(V_1'(0)-V_2'(0))^{-1}\\
&\cdot(r_0(0)+r_1(0)\sqrt{V_1(0)-e_k(h)})^2+\mathcal O(h^{5/2}e^{-2S(e_k)/h}),
\end{align*}
uniformly as $h\to 0_+$ where
$$e_k=e_k(h):=\mathcal A^{-1}((k+\frac{1}{2})\pi h),$$
$$S(e_k)=\int_{b(e_k)}^0\sqrt{V_1(t)-e_k}dt+\int^{c(e_k)}_0\sqrt{V_2(t)-e_k}dt.$$
\end{thm}

\section{Some preliminaries}\label{secondsec}
As in \cite{FMW} let $x_{\infty}>0$ be a sufficiently large number and  $f\in C^{\infty}(\mathbb R_+,\mathbb R_+)$ be a function such that $f(x)=x$ for $x$ large enough, $f(x)=0$ for $x\in [0,x_{\infty}]$ and
$$f'(t)\geq \delta(t^{-1}f(t)+t^{-3}f(t)^3),$$
for some constant $\delta>0$. In the sequel we will use the following notation:
\begin{align*}
&I_b:=[b,0]\ ;\ I_c:=[0,c];\\
&I_L:=(-\infty,b]\  ;\  I_R^{\theta}:=F_{\theta}([c,+\infty)),
\end{align*}
for real $b$ and $c$ where $F_{\theta}(x):=x+i\theta f(x)$. For complex $b$ and $c$, we denote by the same notation appropriate curves in the complex plane connecting the end points.
Then $\theta>0$ small enough, $E\in\mathcal D_I$ and some positive constant $C$, we have
$$\mathrm{Im}\, \int_{x_{\infty}}^{F_{\theta}}\sqrt {E-V_2(t)}dt\geq -Ch,$$
where the integral is taken along $I_R^{\theta}$.

We fix $E'\in I$ and $x_0\in (a(E'),b(E'))$. Then we can define a function $\xi_1(x;E)$ which depends analytically on $x\in \{x\in \mathbb C;\lvert \mathrm{Im}\, (x-b(E'))\rvert<\delta'\langle \mathrm{Re}\, (x-b(E'))\rangle,\ \mathrm{Re}\, x>x_0\}$ with $\delta'>0$ small enough and $E\in \mathcal D_I$ sufficiently close to $E'$ as follows:
$$\xi_1(x;E)=\left(\frac{3}{2}\int_{b(E)}^x\sqrt{V_1(t)-E}dt\right)^{2/3}.$$
Similarly, we can define,
$$\xi_2(x;E)=\left(\frac{3}{2}\int_{c(E)}^x\sqrt{E-V_2(t)}dt\right)^{2/3},$$
which depends analytically on $E\in \mathcal D_I$ and $x\in \{x\in \mathbb C;\lvert \mathrm{Im}\, (x-c(E))\rvert<\delta'\langle \mathrm{Re}\, (x-c(E))\rangle\}$ with $\delta'>0$ small enough.

In the same way as in \cite[Appendix A.2]{FMW} we have solutions to $(P_j-E)u=0$ for $E\in \mathcal D_I$. We denote by $\mathrm{Ai}(x)$ and $\mathrm{Bi}(x)$ the Airy functions.
\begin{pro}\label{yafaev}
Let $E\in \mathcal D_I$. Then,\\
(i) For sufficiently small $h>0$, the equation $(P_1-E)u=0$ admits two solutions $u_{1,R}^{\pm}$ on $\Gamma_1=\{x\in \mathbb C;\lvert\mathrm{Im}\, (x-b(E))\rvert\leq \delta_1\lvert\mathrm{Re}\, (x-b(E))\rvert,\ \mathrm{Re}\, x\geq x_0\}$ with sufficiently small $\delta_1>0$ such that as $x\to +\infty$,
$$u_{1,R}^{\pm}(x)\sim(1+\mathcal O(h))\frac{h^{1/6}}{\sqrt{\pi}}(V_1(x)-E)^{-1/4}e^{\pm\int_{b(E)}^x\sqrt{V_1(t)-E}dt/h},$$
uniformly with respect to $h>0$ small enough, and as $h\to 0_+$,
\begin{align*}
u_{1,R}^-=&2(\xi'_1(x))^{-1/2}\mathrm{Ai}(h^{-2/3}\xi_1(x))(1+\mathcal O(h))\\
&\qquad \qquad \qquad \qquad on\ \Gamma_1\cap \{\mathrm{Re}\, \xi_1(x)\geq 0\};\\
u_{1,R}^-=&2(\xi'_1(x))^{-1/2}\mathrm{Ai}(h^{-2/3}\xi_1(x))+\mathcal O(h(1+h^{-2/3}\lvert \xi_1(x)\rvert)^{-1/4}))\\
&\qquad \qquad \qquad \qquad on\ \Gamma_1\cap \{\mathrm{Re}\, \xi_1(x)\leq 0\};\\
u_{1,R}^+=&(\xi'_1(x))^{-1/2}\mathrm{Bi}(h^{-2/3}\xi_1(x))(1+\mathcal O(h))\\
&\qquad \qquad \qquad \qquad on\ \Gamma_1\cap \{\mathrm{Re}\, \xi_1(x)\geq 0\};\\
u_{1,R}^+=&(\xi'_1(x))^{-1/2}\mathrm{Bi}(h^{-2/3}\xi_1(x))+\mathcal O(h(1+h^{-2/3}\lvert \xi_1(x)\rvert)^{-1/4}))\\
&\qquad \qquad \qquad \qquad on\ \Gamma_1\cap \{\mathrm{Re}\, \xi_1(x)\leq 0\}.
\end{align*}
(ii) For sufficiently small $h>0$, there exist two constants $a_2^{\pm}=1+\mathcal O(h)$ and the solutions $u_{2,L}^{\pm}$ to the equation $(P_2-E)u=0$ on $\Gamma_2=\{x\in \mathbb C;\lvert\mathrm{Im}\, (x-c(E))\rvert\leq \delta_2\lvert\mathrm{Re}\, (x-c(E))\rvert\}$ with sufficiently small $\delta_2>0$ such that,
\begin{align*}
&u_{2,L}^{\pm}(x)\sim(1+\mathcal O(h))\frac{h^{1/6}}{\sqrt{\pi}}(V_2(x)-E)^{-1/4}e^{\mp\int_{c(E)}^x\sqrt{V_2(t)-E}dt/h},\ (x\to -\infty);\\
&e^{\mp i\frac{\pi}{4}}(\frac{1}{2}a_2^-u_{2,L}^-\pm ia_2^+u_{2,L}^+(x))\\
&\quad \sim(1+\mathcal O(h))\frac{h^{1/6}}{\sqrt{\pi}}(E-V_2(x))^{-1/4}e^{\mp i\int_{c(E)}^x\sqrt{E-V_2(t)}dt/h},\ (x\to +\infty),
\end{align*}
uniformly with respect to $h>0$ small enough, and as $h\to 0_+$,
\begin{align*}
u_{2,L}^-=&2(\xi'_2(x))^{-1/2}\check{\mathrm{Ai}}(h^{-2/3}\xi_2(x))+\mathcal O(h(1+h^{-2/3}\lvert \xi_2(x)\rvert)^{-1/4}))\\
&\qquad \qquad \qquad \qquad on\ \Gamma_2\cap \{\mathrm{Re}\, \xi_2(x)\geq 0\};\\
u_{2,L}^-=&2(\xi'_2(x))^{-1/2}\check{\mathrm{Ai}}(h^{-2/3}\xi_2(x))(1+\mathcal O(h))\\
&\qquad \qquad \qquad \qquad on\ \Gamma_2\cap \{\mathrm{Re}\, \xi_2(x)\leq 0\};\\
u_{2,L}^+=&(\xi'_2(x))^{-1/2}\check{\mathrm{Bi}}(h^{-2/3}\xi_2(x))+\mathcal O(h(1+h^{-2/3}\lvert \xi_2(x)\rvert)^{-1/4}))\\
&\qquad \qquad \qquad \qquad on\ \Gamma_2\cap \{\mathrm{Re}\, \xi_2(x)\geq 0\};\\
u_{2,L}^+=&(\xi'_2(x))^{-1/2}\check{\mathrm{Bi}}(h^{-2/3}\xi_2(x))(1+\mathcal O(h))\\
&\qquad \qquad \qquad \qquad on\ \Gamma_2\cap \{\mathrm{Re}\, \xi_2(x)\leq 0\},
\end{align*}
where $\check{\mathrm{Ai}}(x)=\mathrm{Ai}(-x)$ and $\check{\mathrm{Bi}}(x)=\mathrm{Bi}(-x)$.
\end{pro}

\begin{rem}
Similarly, we have two solutions $u_{1,L}^{\pm}$ on $\tilde \Gamma_1=\{x\in \mathbb C;\lvert\mathrm{Im}\, (x-a(E))\rvert\leq \tilde \delta_1\lvert\mathrm{Re}\, (x-a(E))\rvert,\ \mathrm{Re}\, x\leq x_0\}$ with sufficiently small $\tilde \delta_1>0$ with the asymptotic behavior,
$$u_{1,L}^{\pm}(x)\sim(1+\mathcal O(h))\frac{h^{1/6}}{\sqrt{\pi}}(V_1(x)-E)^{-1/4}e^{\mp\int_{a(E)}^x\sqrt{V_1(t)-E}dt/h},\ (x\to -\infty).$$
\end{rem}

As in \cite[Proposition 5.1]{FMW} we have,
\begin{pro}\label{LRC}
Let $u_{1,R}^{\pm}$ and $u_{1,L}^{\pm}$ as above. Then one has
$$u_{1,L}^{\pm}=a_{\pm}u_{1,R}^-+b_{\pm}u_{1,R}^+$$
with,
\begin{align*}
&a_-=\sin\frac{\mathcal A(E)}{h}+\mathcal O(h)\quad;\quad b_-=2\cos\frac{\mathcal A(E)}{h}+\mathcal O(h)\\
&a_+=\frac{1}{2}\cos\frac{\mathcal A(E)}{h}+\mathcal O(h)\quad;\quad b_+=-\sin\frac{\mathcal A(E)}{h}+\mathcal O(h),
\end{align*}
as $h\to 0_+$.
\end{pro}
Setting
$$u_{2,R}^{\pm}:=e^{\mp i\frac{\pi}{4}}(\frac{1}{2}a_2^-u_{2,L}^-\pm ia_2^+u_{2,L}^+),$$
by Proposition \ref{yafaev} we have,
$$u_{2,R}^{\pm}\sim(1+\mathcal O(h))\frac{h^{1/6}}{\sqrt{\pi}}(E-V_2(x))^{-1/4}e^{\mp i\int_{c(E)}^x\sqrt{E-V_2(t)}dt/h},\ (x\to +\infty).$$

The Wronskians $\mathcal W[u_{j,L}^-,u_{j,L}^+]$ and $\mathcal W[u_{j,R}^-,u_{j,R}^+]$ are independent of the variable $x$ and satisfies
\begin{equation}\label{myeq2.0.1}
\begin{split}
&\mathcal W[u_{j,L}^-,u_{j,L}^+]=-\frac{2}{\pi h^{2/3}}(1+\mathcal O(h)),\\
&\mathcal W[u_{1,R}^-,u_{1,R}^+]=\frac{2}{\pi h^{2/3}}(1+\mathcal O(h)),\\
&\mathcal W[u_{2,R}^-,u_{2,R}^+]=-\frac{2i}{\pi h^{2/3}}(1+\mathcal O(h)).
\end{split}
\end{equation}
To see the growth and decay from the points $b$ and $c$, we define the solutions $u_{j,b}^{\pm}$ on $I_L$, $u_{j,c}^{\pm}$ on $I_R^{\theta}$, $v_{j,b}^{\pm}$ on $I_b$ and $v_{j,c}^{\pm}$ on $I_c$ as follows:
\begin{align*}
&u_{1,b}^{\pm}:=u_{1,L}^{\pm},\ u_{2,b}^+:=e^{-S_2/h}u_{2,L}^+,\ u_{2,b}^-:=e^{S_2/h}u_{2,L}^-,\\
&u_{2,c}^{\pm}:=u_{2,R}^{\pm},\ u_{1,c}^+:=e^{-S_1/h}u_{1,R}^+,\ u_{1,c}^-:=e^{S_1/h}u_{1,R}^-,\\
&v_{1,b}^{\pm}:=u_{1,R}^{\pm},\ v_{2,b}^+:=e^{S_2/h}u_{2,L}^-,\ v_{2,b}^-:=e^{-S_2/h}u_{2,L}^+,\\
&v_{2,c}^{\pm}:=u_{2,L}^{\pm},\ v_{1,c}^+:=e^{S_1/h}u_{1,R}^-,\ v_{1,c}^-:=e^{-S_1/h}u_{1,R}^+,
\end{align*}
where $S_1:=\int_b^c\sqrt{V_1(t)-E}dt,\ S_2:=\int_b^c\sqrt{V_2(t)-E}dt$.

\section{Fundamental solutions}\label{section2.1}
In this section we introduce fundamental solutions used to construct solutions to the system.
\subsection{Fundamental solutions on $I_b$}\mbox{} \\
We define fundamental solutions
$$K_{1,b},\ K_{1,b}',\ K_{1,b}'': C(I_b)\to C^2(I_b),$$
of $P_1-E$ and
$$K_{2,b}: C(I_b)\to C^2(I_b),$$
of $P_2-E$ by setting for $v\in C(I_b)$,
\begin{align*}
&K_{1,b}[v](x):=\frac{1}{h^2\mathcal W[v_{1,b}^-,v_{1,b}^+]}\left(v_{1,b}^-(x)\int^x_bv_{1,b}^+(t)v(t)dt+v_{1,b}^+(x)\int^0_xv_{1,b}^-(t)v(t)dt\right),\\
&K_{1,b}'[v](x):=\frac{1}{h^2\mathcal W[v_{1,b}^-,v_{1,b}^+]}\left(v_{1,b}^-(x)\int^x_bv_{1,b}^+(t)v(t)dt-v_{1,b}^+(x)\int^x_bv_{1,b}^-(t)v(t)dt\right),\\
&K_{1,b}''[v](x):=\frac{1}{h^2\mathcal W[v_{1,b}^-,v_{1,b}^+]}\left(-v_{1,b}^-(x)\int^0_xv_{1,b}^+(t)v(t)dt+v_{1,b}^+(x)\int^0_xv_{1,b}^-(t)v(t)dt\right),\\
&K_{2,b}[v](x):=\frac{1}{h^2\mathcal W[v_{2,b}^-,v_{2,b}^+]}\left(v_{2,b}^-(x)\int^x_bv_{2,b}^+(t)v(t)dt+v_{2,b}^+(x)\int^0_xv_{2,b}^-(t)v(t)dt\right),\\
\end{align*}
where $\mathcal W[v_{1,b}^-,v_{1,b}^+]$ is the Wronskian of $v_{1,b}^-$ and $v_{1,b}^+$ and so on.
Then $K_{1,b},\ K_{1,b}'$ and $K_{1,b}''$ satisfy
$$(P_1-E)K_{1,b}={\bf 1},\ (P_1-E)K_{1,b}'={\bf 1},\ (P_1-E)K_{1,b}''={\bf 1},$$
and $K_{2,b}$ satisfies
$$(P_2-E)K_{2,b}={\bf 1}.$$
For an interval $I$ and sufficiently small $h_0>0$, let $v(x,h)\in C(I)$ be a family of functions satisfying $v(x,h)\neq0$ for $x\in I$, $0<h<h_0$. For $h\in (0,h_0]$ we define $C(v,h)$ as the set of continuous functions $u$ on $I$ equipped with the norm $\lVert u\rVert_{C(v,h)}:=\sup_{x\in I}\lvert u(x)\rvert\lvert v(x,h)\rvert^{-1}$. In the following, we consider families of functions $u(x,h)\in C(v,h)$, $h\in (0,h_0]$ and fimilies of operators $A(h)\in \mathcal L(C(v,h))$, $h\in (0,h_0]$.

Note that $v_{j,b}^{\pm}\neq 0$ and $v_{j,c}^{\pm}\neq 0$ on $I_b$ and $I_c$ respectively for sufficiently small $h$. In view of the construction of solutions to the system, we prove,

\begin{lem}\label{fundamental}
As $h$ goes to $0_+$ one has,
\begin{align}
&\lVert hK_{1,b}W\rVert_{\mathcal L(C(v_{1,b}^{\pm},h))}=\mathcal O(1), \label{myeq2.1}\\
&\lVert hK_{2,b}W^*\rVert_{\mathcal L(C(v_{1,b}^{\pm},h))}=\mathcal O(h^{1/2}), \label{myeq2.2}\\
&\lVert hK_{1,b}'W\rVert_{\mathcal L(C(v_{2,b}^{+},h))}=\mathcal O(h^{1/2}), \label{myeq2.4}\\
&\lVert hK_{1,b}''W\rVert_{\mathcal L(C(v_{2,b}^{-},h))}=\mathcal O(h^{1/2}), \label{myeq2.3}\\
&\lVert hK_{2,b}W^*\rVert_{\mathcal L(C(v_{2,b}^{\pm},h))}=\mathcal O(1). \label{myeq2.5}
\end{align}
Moreover, there exist complex numbers $\alpha_{j,b},\beta_{1,b}^{\pm},\beta_{2,b}$ depending on $v$ such that,
\begin{align}
&\begin{pmatrix}hK_{1,b}Wv(b)\\  h\partial_x (K_{1,b}Wv)(b)\end{pmatrix}=\alpha_{1,b}\begin{pmatrix}v_{1,b}^{+}(b)\\ \partial_xv_{1,b}^{+}(b)\end{pmatrix}, \label{myeq2.6}\\
&\begin{pmatrix}hK_{2,b}W^*v(b)\\  h\partial_x (K_{2,b}W^*v)(b)\end{pmatrix}=\alpha_{2,b}\begin{pmatrix}v_{2,b}^{+}(b)\\ \partial_xv_{2,b}^{+}(b)\end{pmatrix}, \label{myeq2.7}\\
&\begin{pmatrix}hK_{1,b}'Wv(b)\\  h\partial_x (K_{1,b}'Wv)(b)\end{pmatrix}=\begin{pmatrix}0\\ 0\end{pmatrix}, \label{myeq2.9}\\
&\begin{pmatrix}hK_{1,b}''Wv(b)\\  h\partial_x (K_{1,b}''Wv)(b)\end{pmatrix}=\beta_{1,b}^-\begin{pmatrix}v_{1,b}^{-}(b)\\ \partial_xv_{1,b}^{-}(b)\end{pmatrix}+\beta_{1,b}^+\begin{pmatrix}v_{1,b}^{+}(b)\\ \partial_xv_{1,b}^{+}(b)\end{pmatrix}, \label{myeq2.8}\\
&\begin{pmatrix}hK_{2,b}W^*v(b)\\  h\partial_x (K_{2,b}W^*v)(b)\end{pmatrix}=\beta_{2,b}\begin{pmatrix}v_{2,b}^{+}(b)\\ \partial_xv_{2,b}^{+}(b)\end{pmatrix}, \label{myeq2.10}
\end{align}
and
$$\lvert \alpha_{1,b}\rvert\leq C\lVert v\rVert_{C(v_{1,b}^{\pm},h)},\ \lvert \alpha_{2,b}\rvert\leq Ch^{5/6}\lVert v\rVert_{C(v_{1,b}^{\pm},h)},\ \lvert \beta_{1,b}^{\pm}\rvert\leq Ch^{5/6}\lVert v\rVert_{C(v_{2,b}^{-},h)},\ \lvert \beta_{2,b}\rvert\leq C\lVert v\rVert_{C(v_{2,b}^{\pm},h)}.$$
\end{lem}
\begin{proof}
(i) First, we shall prove \eqref{myeq2.1} and \eqref{myeq2.6}.
We set,
\begin{align*}
&\tilde v_{1,b}^{\pm}:=\max\{\lvert v_{1,b}^{\pm}\rvert, \lvert h\partial_x v_{1,b}^{\pm}\rvert\},\\
&\mathcal V_{1,b}(x,t):=\tilde v_{1,b}^-(x)\tilde v_{1,b}^+(t){\bf 1}_{\{t<x\}}+\tilde v_{1,b}^+(x)\tilde v_{1,b}^-(t){\bf 1}_{\{t>x\}}.
\end{align*}
By an integration by parts we have,
$$
\lvert hK_{1,b}Wv(x)\rvert=\mathcal O(h^{-1/3})\Bigg( \int_b^0\mathcal V_{1,b}(x,t)\lvert v(t)\rvert dt+h\mathcal V_{1,b}(x,0)\lvert v(0)\rvert+h\mathcal V_{1,b}(x,b)\lvert v(b)\rvert\Bigg ).
$$
Using the asymptotics of $v_{1,b}^{\pm}$ and $h\partial_xv_{1,b}^{\pm}$ on $I_b$ and fixing some constant $C_1>0$ sufficiently large we obtain,
\begin{itemize}
\item If $b\leq x,t\leq b+C_1h^{2/3}$, then,
$$\mathcal V_{1,b}(x,t)\lvert v_{1,b}^{\pm}(t)\rvert\lvert v_{1,b}^{\pm}(x)\rvert^{-1}=\mathcal O(1).$$
\item If $b\leq x\leq b+C_1h^{2/3}\leq t\leq 0$, then,
$$\mathcal V_{1,b}(x,t)\lvert v_{1,b}^{\pm}(t)\rvert\lvert v_{1,b}^{\pm}(x)\rvert^{-1}=\mathcal O(h^{1/3})\frac{e^{-\mathrm{Re}\int_b^t(V_1-E)^{1/2}/h\pm \mathrm{Re}\int_b^t(V_1-E)^{1/2}/h}}{\lvert V_1(t)-E\rvert^{1/2}}.$$
\item If $b\leq t\leq b+C_1h^{2/3}\leq x\leq 0$, then,
$$\mathcal V_{1,b}(x,t)\lvert v_{1,b}^{\pm}(t)\rvert\lvert v_{1,b}^{\pm}(x)\rvert^{-1}=\mathcal O(1)e^{-\mathrm{Re}\int_b^x(V_1-E)^{1/2}/h\mp \mathrm{Re}\int_b^x(V_1-E)^{1/2}/h}.$$
\item If $b+C_1h^{2/3}\leq x,t\leq 0$, then,
$$\mathcal V_{1,b}(x,t)\lvert v_{1,b}^{\pm}(t)\rvert\lvert v_{1,b}^{\pm}(x)\rvert^{-1}=\mathcal O(h^{1/3})\frac{e^{-\lvert\mathrm{Re}\int_t^x(V_1-E)^{1/2}\rvert/h\pm \mathrm{Re}\int_x^t(V_1-E)^{1/2}/h}}{\lvert V_1(t)-E\rvert^{1/2}}.$$
\end{itemize}
Hence, observing $\lvert V_1(t)-E\rvert^{1/4}\geq C\lvert t-b\rvert\geq Ch^{1/6}$ for $b+C_1h^{2/3}\leq t\leq 0$, we have
$$\mathcal V_{1,b}(x,t)\lvert v_{1,b}^{\pm}(t)\rvert\lvert v_{1,b}^{\pm}(x)\rvert^{-1}=\mathcal O(1)\ \mathrm{as}\ h\to 0_+,$$
uniformly with respect to $b\leq x,t\leq 0$. Thus we obtain
$$\lVert \mathcal V_{1,b}(x,0)\lvert v(0)\rvert\rVert_{C(v_{1,b}^{\pm},h)},\ \lVert \mathcal V_{1,b}(x,b)\lvert v(b)\rvert\rVert_{C(v_{1,b}^{\pm},h)}\leq C\lVert v\rVert_{C(v_{1,b}^{\pm},h)}.$$
Moreover, regardless of $b\leq x\leq b+C_1h^{2/3}$ or $b+C_1h^{2/3}<x\leq 0$, we have
\begin{align*}
\int_b^0\mathcal V_{1,b}(x,t)\lvert v_{1,b}^{\pm}(t)\rvert\lvert v_{1,b}^{\pm}(x)\rvert^{-1} dt&= \mathcal O(1)\int_b^{b+C_1h^{2/3}} dt+\mathcal O(h^{1/3})\int_{b+C_1h^{2/3}}^0( t-b)^{-1/2} dt\\
&=\mathcal O(h^{1/3}).
\end{align*}
Hence we obtain
$$\left \lVert \int_b^0\mathcal V_{1,b}(x,t)\lvert v(t)\rvert dt\right \rVert_{C(v_{1,b}^{\pm},h)}\leq Ch^{1/3}\lVert v\rVert_{C(v_{1,b}^{\pm},h)},$$
which completes the proof of \eqref{myeq2.1}. Observing
\begin{align*}
&hK_{1,b}Wv(b)=\frac{1}{h^2\mathcal W[v_{1,b}^-,v_{1,b}^+]}v_{1,b}^+(b)\int^0_bv_{1,b}^-(t)v(t)dt,\\
&h\partial_xK_{1,b}Wv(b)=\frac{1}{h^2\mathcal W[v_{1,b}^-,v_{1,b}^+]}\partial_xv_{1,b}^+(b)\int^0_bv_{1,b}^-(t)v(t)dt,\\
\end{align*}
we estimate $\lvert (hK_{1,b}Wv(b))(v_{1,b}^+(b))^{-1}\rvert$ and $\lvert (h\partial_xK_{1,b}Wv(b))(\partial_x v_{1,b}^+(b))^{-1}\rvert$ by the similar calculation as above and obtain \eqref{myeq2.6}.

(ii) We shall prove \eqref{myeq2.2} and \eqref{myeq2.7}.
We set,
\begin{align*}
&\tilde v_{2,b}^{\pm}:=\max\{\lvert v_{2,b}^{\pm}\rvert, \lvert h\partial_x v_{2,b}^{\pm}\rvert\},\\
&\mathcal V_{2,b}(x,t):=\tilde v_{2,b}^-(x)\tilde v_{2,b}^+(t){\bf 1}_{\{t<x\}}+\tilde v_{2,b}^+(x)\tilde v_{2,b}^-(t){\bf 1}_{\{t>x\}}.
\end{align*}
By an integration by parts we have,
\begin{equation}\label{myeq2.11}
\lvert hK_{2,b}W^*v(x)\rvert=\mathcal O(h^{-1/3})\Bigg( \int_b^0\mathcal V_{2,b}(x,t)\lvert v(t)\rvert dt+h\mathcal V_{2,b}(x,0)\lvert v(0)\rvert+h\mathcal V_{2,b}(x,b)\lvert v(b)\rvert\Bigg ).
\end{equation}
Using the asymptotics of $v_{j,b}^{\pm}$ and $h\partial_xv_{2,b}^{\pm}$ on $I_b$ and fixing some constant $C_1>0$ sufficiently large we obtain,
\begin{itemize}
\item If $b\leq x,t\leq b+C_1h^{2/3}$, then,
$$\mathcal V_{2,b}(x,t)\lvert v_{1,b}^{\pm}(t)\rvert\lvert v_{1,b}^{\pm}(x)\rvert^{-1}=\mathcal O(h^{1/3})e^{-\lvert\mathrm{Re}\int_t^x(V_2-E)^{1/2}\rvert/h}.$$
\item If $b\leq x\leq b+C_1h^{2/3}\leq t\leq 0$, then,
$$\mathcal V_{2,b}(x,t)\lvert v_{1,b}^{\pm}(t)\rvert\lvert v_{1,b}^{\pm}(x)\rvert^{-1}=\mathcal O(h^{1/2})\frac{e^{-\lvert\mathrm{Re}\int_t^x(V_2-E)^{1/2}\rvert/h\pm \mathrm{Re}\int_b^t(V_1-E)^{1/2}/h}}{\lvert V_1(t)-E\rvert^{1/4}}.$$
\item If $b\leq t\leq b+C_1h^{2/3}\leq x\leq 0$, then,
$$\mathcal V_{2,b}(x,t)\lvert v_{1,b}^{\pm}(t)\rvert\lvert v_{1,b}^{\pm}(x)\rvert^{-1}=\mathcal O(h^{1/6})e^{-\lvert\mathrm{Re}\int_t^x(V_2-E)^{1/2}\rvert/h\mp \mathrm{Re}\int_b^x(V_1-E)^{1/2}/h}.$$
\item If $b+C_1h^{2/3}\leq x,t\leq 0$, then,
$$\mathcal V_{2,b}(x,t)\lvert v_{1,b}^{\pm}(t)\rvert\lvert v_{1,b}^{\pm}(x)\rvert^{-1}=\mathcal O(h^{1/3})\frac{e^{-\lvert\mathrm{Re}\int_t^x(V_2-E)^{1/2}\rvert/h\pm \mathrm{Re}\int_x^t(V_1-E)^{1/2}/h}}{\lvert V_1(t)-E\rvert^{1/4}}.$$
\end{itemize}
Hence, observing that $\int_b^{b+C_1h^{2/3}}(V_1(t)-E)^{1/2}dt=O(h)$, we have
$$\mathcal V_{2,b}(x,t)\lvert v_{1,b}^{\pm}(t)\rvert\lvert v_{1,b}^{\pm}(x)\rvert^{-1}=\mathcal O(h^{1/6}),$$
uniformly with respect ot $b\leq x,t\leq 0$. Thus we have
$$\lVert \mathcal V_{2,b}(x,0)\lvert v(0)\rvert\rVert_{C(v_{1,b}^{\pm},h)},\ \lVert \mathcal V_{2,b}(x,b)\lvert v(b)\rvert\rVert_{C(v_{1,b}^{\pm},h)}\leq Ch^{1/6}\lVert v\rVert_{C(v_{1,b}^{\pm},h)}.$$

When $b\leq x\leq b+C_1h^{2/3}$, there exist constants $\delta>0$ small enough and $\alpha>0$ such that
$$\int_b^0\mathcal V_{2,b}(x,t)\lvert v_{1,b}^{\pm}(t)\rvert\lvert v_{1,b}^{\pm}(x)\rvert^{-1} dt= \mathcal O(h^{1/3})\int_b^{b+\delta}e^{-\alpha\lvert x-t\rvert/h} dt+\mathcal O(e^{-\alpha/h})=\mathcal O(h^{4/3}).$$
On the other hand, if $b+C_1h^{2/3}<x\leq-\delta'$ with $>\delta'>0$ small enough there exist a constant $\alpha>0$ such that
$$\int_b^0\mathcal V_{2,b}(x,t)\lvert v_{1,b}^{\pm}(t)\rvert\lvert v_{1,b}^{\pm}(x)\rvert^{-1} dt= \mathcal O(h^{1/6})\int_b^{-\delta'/2}e^{-\alpha\lvert x-t\rvert/h} dt+\mathcal O(e^{-\alpha/h})=\mathcal O(h^{7/6}).$$
Finally, when $-\delta'<x\leq 0$ there exist constants $\alpha,\beta>0$ such that
$$\int_b^0\mathcal V_{2,b}(x,t)\lvert v_{1,b}^{\pm}(t)\rvert\lvert v_{1,b}^{\pm}(x)\rvert^{-1} dt= \mathcal O(h^{1/3})\int_{-2\delta'}^0e^{-\beta\lvert x^2-t^2\rvert/h}dt+\mathcal O(e^{-\alpha/h})=\mathcal O(h^{5/6}).$$
Here we used $\int_{-2\delta'}^0e^{-\beta\lvert x^2-t^2\rvert/h}dt=\mathcal O(h^{1/2})$. To prove this estimate we write
$$\int_{-2\delta'}^0e^{-\beta\lvert x^2-t^2\rvert/h}dt=\int^x_{-2\delta'}e^{-\beta\lvert x^2-t^2\rvert/h}dt+\int^0_{x}e^{-\beta\lvert x^2-t^2\rvert/h}dt.$$
The first term is estimated as follows;
$$\int^x_{-2\delta'}e^{-\beta\lvert x^2-t^2\rvert/h}dt=\int^0_{-2\delta'-x}e^{-\beta\lvert x^2-(t+x)^2\rvert/h}dt\leq\int^0_{-2\delta'-x}e^{-\beta t^2/h}dt=\mathcal O(h^{1/2}),$$
and the second term is estimated as follows;
\begin{align*}
\int^0_{x}e^{-\beta\lvert x^2-t^2\rvert/h}dt&=\int^{-x}_{0}e^{-\beta\lvert x^2-(t+x)^2\rvert/h}dt=\int^{-x}_{0}e^{-\beta\lvert 2(-x)t-t^2\rvert/h}dt\leq\int^{-x}_{0}e^{-\beta\lvert (-x)t\rvert/h}dt\\
&\leq\int^{-x}_{0}e^{-\beta t^2/h}dt=\mathcal O(h^{1/2}),
\end{align*}
where we used that $(-x)t\geq t^2$ for $0\leq t\leq -x$ in the inequalities.
Hence, we obtain
$$\left\lVert\int_b^0\mathcal V_{2,b}(x,t)\lvert v(t)\rvert dt\right\rVert_{C(v_{1,b}^{\pm},h)}\leq Ch^{5/6}\lVert v\rVert_{C(v_{1,b}^{\pm},h)},$$
which completes the proof of \eqref{myeq2.2}. Noting that there exists a constant $C>0$ such that $\lvert v_{1,b}^{\pm}(b)\rvert\leq C$ and $C^{-1}h^{1/6}\leq \lvert v_{2,b}^{-}(b)\rvert$, we estimate $\lvert (hK_{2,b}W^*v(b))(v_{2,b}^+(b))^{-1}\rvert$ and $\lvert (h\partial_xK_{2,b}W~*v(b))(\partial_x v_{2,b}^+(b))^{-1}\rvert$ by the similar calculation as above and obtain \eqref{myeq2.7}.

(iii) We shall prove \eqref{myeq2.4} and \eqref{myeq2.9}.
We set,
$$\mathcal V_{1,b}'(x,t):=\tilde v_{1,b}^-(x)\tilde v_{1,b}^+(t)+\tilde v_{1,b}^+(x)\tilde v_{1,b}^-(t).$$
By an integration by parts we have,
$$\lvert hK_{1,b}'Wv(x)\rvert= \mathcal O(h^{-1/3})\Bigg( \int_b^x\mathcal V_{1,b}'(x,t)\lvert v(t)\rvert dt+h\mathcal V_{1,b}''(x,x)\lvert v(x)\rvert+h\mathcal V_{1,b}'(x,b)\lvert v(b)\rvert\Bigg ).$$
We obtain,
\begin{itemize}
\item If $b\leq t\leq x\leq b+C_1h^{2/3}$, then,
$$\mathcal V_{1,b}'(x,t)\lvert v_{2,b}^{+}(t)\rvert\lvert v_{2,b}^{+}(x)\rvert^{-1}=\mathcal O(1)e^{-\mathrm{Re}\int_t^x(V_2-E)^{1/2}/h}.$$
\item If $b\leq t\leq b+C_1h^{2/3}\leq x\leq 0$, then,
$$\mathcal V_{1,b}'(x,t)\lvert v_{2,b}^{+}(t)\rvert\lvert v_{2,b}^{+}(x)\rvert^{-1}=\mathcal O(h^{1/6})\frac{e^{-\mathrm{Re}\int_t^x(V_2-E)^{1/2}/h+ \mathrm{Re}\int_b^x(V_1-E)^{1/2}/h}}{\lvert V_1(x)-E\rvert^{1/4}}.$$
\item If $b+C_1h^{2/3}\leq t\leq x\leq 0$, then,
$$\mathcal V_{1,b}'(x,t)\lvert v_{2,b}^{+}(t)\rvert\lvert v_{2,b}^{+}(x)\rvert^{-1}=\mathcal O(h^{1/3})\frac{e^{-\lvert\mathrm{Re}\int_t^x(V_2-E)^{1/2}\rvert/h+ \mathrm{Re}\int_t^x(V_1-E)^{1/2}/h}}{\lvert V_1(t)-E\rvert^{1/4}\lvert V_1(x)-E\rvert^{1/4}}.$$
\end{itemize}
Hence, we have $\mathcal V_{1,b}'(x,t)\lvert v_{2,b}^{+}(t)\rvert\lvert v_{2,b}^{+}(x)\rvert^{-1}=\mathcal O(1)$ as $h\to 0_+$. Thus we have
$$\lVert \mathcal V_{1,b}'(x,b)\lvert v(b)\rvert\rVert_{C(v_{2,b}^{+},h)},\ \lVert \mathcal V_{1,b}'(x,x)\lvert v(x)\rvert\rVert_{C(v_{2,b}^{-},h)}\leq C\lVert v\rVert_{C(v_{2,b}^{+},h)}.$$
When $b\leq x\leq b+C_1h^{2/3}$ or $b+C_1h^{2/3}<x\leq-\delta'$ with $\delta'>0$ small enough, there exist a constant $\alpha>0$ such that
$$
\int_b^x\mathcal V_{1,b}'(x,t)\lvert v_{2,b}^{+}(t)\rvert\lvert v_{2,b}^{+}(x)\rvert^{-1}dt\leq \mathcal O(1)\int_b^{x}e^{-\alpha(t-x)/h}dt=\mathcal O(h).$$
When $-\delta'<x\leq 0$ there exist constants $\alpha,\beta>0$ such that
$$\int_b^x\mathcal V_{1,b}'(x,t)\lvert v_{2,b}^{+}(t)\rvert\lvert v_{2,b}^{+}(x)\rvert^{-1} dt=\mathcal O(h^{1/3})\int_{-2\delta'}^xe^{-\beta\lvert x^2-t^2\rvert/h}dt+\mathcal O(e^{-\alpha/h})=\mathcal O(h^{5/6}).$$
Hence, we obtain
$$\left\lVert\int_b^x\mathcal V_{1,b}'(x,t)\lvert v(t)\rvert dt\right\rVert_{C(v_{2,b}^{+},h)}\leq Ch^{5/6}\lVert v\rVert_{C(v_{2,b}^{+},h)},$$
which completes the proof of \eqref{myeq2.4}. The equation \eqref{myeq2.9} is obvious from the definition of $K_{1,b}'$.

(iv) We shall prove \eqref{myeq2.3} and \eqref{myeq2.8}.
We set,
$$\mathcal V_{1,b}''(x,t)=\mathcal V_{1,b}'(x,t).$$
By an integration by parts we have,
$$\lvert hK_{1,b}''Wv(x)\rvert= \mathcal O(h^{-1/3})\Bigg( \int_x^0\mathcal V_{1,b}''(x,t)\lvert v(t)\rvert dt+h\mathcal V_{1,b}''(x,0)\lvert v(0)\rvert\\
+h\mathcal V_{1,b}''(x,x)\lvert v(x)\rvert\Bigg ).$$
We obtain,
\begin{itemize}
\item If $b\leq x\leq t\leq b+C_1h^{2/3}$, then,
$$\mathcal V_{1,b}''(x,t)\lvert v_{2,b}^{-}(t)\rvert\lvert v_{2,b}^{-}(x)\rvert^{-1}=\mathcal O(1)e^{-\mathrm{Re}\int_x^t(V_2-E)^{1/2}/h}.$$
\item If $b\leq x\leq b+C_1h^{2/3}\leq t\leq 0$, then,
$$\mathcal V_{1,b}''(x,t)\lvert v_{2,b}^{-}(t)\rvert\lvert v_{2,b}^{-}(x)\rvert^{-1}=\mathcal O(h^{1/6})\frac{e^{-\mathrm{Re}\int_x^t(V_2-E)^{1/2}/h+ \mathrm{Re}\int_b^t(V_1-E)^{1/2}/h}}{\lvert V_1(t)-E\rvert^{1/4}}.$$
\item If $b+C_1h^{2/3}\leq x\leq t\leq 0$, then,
$$\mathcal V_{1,b}''(x,t)\lvert v_{2,b}^{-}(t)\rvert\lvert v_{2,b}^{-}(x)\rvert^{-1}=\mathcal O(h^{1/3})\frac{e^{-\lvert\mathrm{Re}\int_x^t(V_2-E)^{1/2}\rvert/h+ \mathrm{Re}\int_x^t(V_1-E)^{1/2}/h}}{\lvert V_1(t)-E\rvert^{1/4}\lvert V_1(x)-E\rvert^{1/4}}.$$
\end{itemize}
Hence, we have $\mathcal V_{1,b}''(x,t)\lvert v_{2,b}^{-}(t)\rvert\lvert v_{2,b}^{-}(x)\rvert^{-1}=\mathcal O(1)$ as $h\to 0_+$. Thus we have
$$\lVert \mathcal V_{1,b}''(x,0)\lvert v(0)\rvert\rVert_{C(v_{2,b}^{-},h)},\ \lVert \mathcal V_{1,b}''(x,x)\lvert v(x)\rvert\rVert_{C(v_{2,b}^{-},h)}\leq C\lVert v\rVert_{C(v_{2,b}^{-},h)}.$$
When $b\leq x\leq b+C_1h^{2/3}$, there exist constants $\delta>0$ small enough and $\alpha>0$ such that
$$\int_x^0\mathcal V_{1,b}''(x,t)\lvert v_{2,b}^{-}(t)\rvert\lvert v_{2,b}^{-}(x)\rvert^{-1} dt=\mathcal O(1)\int_x^{b+\delta}e^{-\alpha(t-x)/h} dt+\mathcal O(e^{-\alpha/h})\\
=\mathcal O(h).$$
On the other hand, if $b+C_1h^{2/3}<x\leq-\delta'$ with $\delta'>0$ small enough, there exist a constant $\alpha$ such that
$$\int_x^0\mathcal V_{1,b}''(x,t)\lvert v_{2,b}^{-}(t)\rvert\lvert v_{2,b}^{-}(x)\rvert^{-1} dt=\mathcal O(1)\int_x^{-\delta'/2}e^{-\alpha(t-x)/h} dt+\mathcal O(e^{-\alpha/h})=\mathcal O(h).$$
Finally, when $-\delta'<x\leq 0$ there exist constants $\alpha,\beta>0$ such that
$$\int_x^0\mathcal V_{1,b}''(x,t)\lvert v_{2,b}^{-}(t)\rvert\lvert v_{2,b}^{-}(x)\rvert^{-1} dt=\mathcal O(h^{1/3})\int_{x}^0e^{-\beta\lvert x^2-t^2\rvert/h}dt+\mathcal O(e^{-\alpha/h})=\mathcal O(h^{5/6}).$$
Hence, we obtain
$$\left\lVert\int_b^0\mathcal V_{1,b}''(x,t)\lvert v(t)\rvert dt\right\rVert_{C(v_{2,b}^{-},h)}\leq Ch^{5/6}\lVert v\rVert_{C(v_{2,b}^{-},h)},$$
which completes the proof of \eqref{myeq2.3}. Noting that there exists a constant $C>0$ such that $C\leq \lvert v_{1,b}^{\pm}(b)\rvert$ $\lvert v_{2,b}^{-}(b)\rvert\leq Ch^{1/6}$, we obtain \eqref{myeq2.8} by the similar calculation as above.

(v) Finally, we shall prove \eqref{myeq2.5} and \eqref{myeq2.10}.
We shall estimate the terms in \eqref{myeq2.11}
We obtain for any $b\leq x,t\leq 0$,
$$\mathcal V_{2,b}(x,t)\lvert v_{2,b}^{\pm}(t)\rvert\lvert v_{2,b}^{\pm}(x)\rvert^{-1}=\mathcal O(h^{1/3})e^{-\lvert\mathrm{Re}\int_t^x(V_2-E)^{1/2}\rvert/h\pm\mathrm{Re}\int_x^t(V_2-E)^{1/2}/h}.$$
Hence, we have $\mathcal V_{2,b}(x,t)\lvert v_{2,b}^{\pm}(t)\rvert\lvert v_{2,b}^{\pm}(x)\rvert^{-1}=\mathcal O(h^{1/3})$ as $h\to 0_+$. Thus we have
$$\lVert \mathcal V_{2,b}(x,b)\lvert v(b)\rvert\rVert_{C(v_{2,b}^{\pm},h)},\ \lVert \mathcal V_{2,b}(x,x)\lvert v(x)\rvert\rVert_{C(v_{2,b}^{\pm},h)}\leq C\lVert v\rVert_{C(v_{2,b}^{\pm},h)}.$$
We have for any $b\leq x\leq 0$,
$$\int_b^0\mathcal V_{2,b}(x,t)\lvert v_{2,b}^{\pm}(t)\rvert\lvert v_{2,b}^{\pm}(x)\rvert^{-1} dt=\mathcal O(h^{1/3})\int_{b}^0dt=\mathcal O(h^{1/3}).$$
Noting that there exists a constant $C>0$ such that $C^{-1}h^{1/6}\leq v_{2,b}^{\pm}(b)\leq Ch^{1/6}$, we obtain \eqref{myeq2.10} by the similar calculation as above.
\end{proof}

\subsection{Fundamental solutions on $I_c$}\mbox{} \\
We define the fundamental solutions of $P_j-E$ on $I_c$ as,
\begin{align*}
&K_{1,c}[v](x):=\frac{1}{h^2\mathcal W[v_{1,c}^+,v_{1,c}^-]}\left(v_{1,c}^+(x)\int^x_0v_{1,c}^-(t)v(t)dt+v_{1,c}^-(x)\int^c_xv_{1,c}^+(t)v(t)dt\right),\\
&K_{2,c}[v](x):=\frac{1}{h^2\mathcal W[v_{2,c}^+,v_{2,c}^-]}\left(v_{2,c}^+(x)\int^x_0v_{2,c}^-(t)v(t)dt+v_{2,c}^-(x)\int^c_xv_{2,c}^+(t)v(t)dt\right),\\
&K_{2,c}'[v](x):=\frac{1}{h^2\mathcal W[v_{2,c}^+,v_{2,c}^-]}\left(-v_{2,c}^+(x)\int^c_xv_{2,c}^-(t)v(t)dt+v_{2,c}^-(x)\int^c_xv_{2,c}^+(t)v(t)dt\right),\\
&K_{2,c}''[v](x):=\frac{1}{h^2\mathcal W[v_{2,c}^+,v_{2,c}^-]}\left(v_{2,c}^+(x)\int^x_0v_{2,c}^-(t)v(t)dt-v_{2,c}^-(x)\int^x_0v_{2,c}^+(t)v(t)dt\right),
\end{align*}
for $v\in C(I_c)$.

Then, one can prove exactly as for Lemma \ref{fundamental} that we have,
\begin{lem}\label{fundamental2}
As $h$ goes to $0_+$ one has,
\begin{align*}
&\lVert hK_{2,c}W^*\rVert_{\mathcal L(C(v_{2,c}^{\pm},h))}=\mathcal O(1),\\
&\lVert hK_{1,c}W\rVert_{\mathcal L(C(v_{2,c}^{\pm},h))}=\mathcal O(h^{1/2}),\\
&\lVert hK_{2,c}'W^*\rVert_{\mathcal L(C(v_{1,c}^{+},h))}=\mathcal O(h^{1/2}),\\
&\lVert hK_{2,c}''W^*\rVert_{\mathcal L(C(v_{1,c}^{-},h))}=\mathcal O(h^{1/2}),\\
&\lVert hK_{1,c}W\rVert_{\mathcal L(C(v_{1,c}^{\pm},h))}=\mathcal O(1).
\end{align*}
Moreover, there exist complex numbers $\alpha_{j,c},\beta_{1,c},\beta_{2,c}^{\pm}$ depending on $v$ such that,
\begin{align*}
&\begin{pmatrix}hK_{2,c}W^*v(c)\\  h\partial_x (K_{2,c}W^*v)(c)\end{pmatrix}=\alpha_{2,c}\begin{pmatrix}v_{2,c}^{+}(c)\\ \partial_xv_{2,c}^{+}(c)\end{pmatrix}, \\
&\begin{pmatrix}hK_{1,c}Wv(c)\\  h\partial_x (K_{1,c}Wv)(c)\end{pmatrix}=\alpha_{1,c}\begin{pmatrix}v_{1,c}^{+}(c)\\ \partial_xv_{1,c}^{+}(c)\end{pmatrix}, \\
&\begin{pmatrix}hK_{2,c}'W^*v(c)\\  h\partial_x (K_{2,c}'W^*v)(c)\end{pmatrix}=\begin{pmatrix}0\\ 0\end{pmatrix}, \\
&\begin{pmatrix}hK_{2,c}''W^*v(c)\\  h\partial_x (K_{2,c}''W^*v)(c)\end{pmatrix}=\beta_{2,c}^-\begin{pmatrix}v_{2,c}^{-}(c)\\ \partial_xv_{2,c}^{-}(c)\end{pmatrix}+\beta_{2,c}^+\begin{pmatrix}v_{2,c}^{+}(c)\\ \partial_xv_{2,c}^{+}(c)\end{pmatrix}, \\
&\begin{pmatrix}hK_{1,c}Wv(c)\\  h\partial_x (K_{1,c}Wv)(c)\end{pmatrix}=\beta_{1,c}\begin{pmatrix}v_{1,c}^{+}(c)\\ \partial_xv_{1,c}^{+}(c)\end{pmatrix}, 
\end{align*}
and
$$\lvert \alpha_{2,c}\rvert\leq C\lVert v\rVert_{C(v_{2,c}^{\pm},h)},\ \lvert \alpha_{1,c}\rvert\leq Ch^{5/6}\lVert v\rVert_{C(v_{2,c}^{\pm},h)},\ \lvert \beta_{2,c}^{\pm}\rvert\leq Ch^{5/6}\lVert v\rVert_{C(v_{1,c}^{-},h)},\ \lvert \beta_{1,c}\rvert\leq C\lVert v\rVert_{C(v_{1,c}^{\pm},h)}.$$
\end{lem}

\subsection{Fundamental solutions on $I_L$ and $I_R^{\theta}$}\mbox{} \\
For any $k\in \mathbb N$ we set
\begin{align*}
C_b^0(I_L):=\{u\,:\, I_L\to \mathbb C \mathrm{\ of\ class\ } C^k;\ \sum_{0\leq j\leq k}\sup_{x\in I_L}\lvert u^{(j)}(x)\rvert\leq +\infty\},\\
C_b^0(I_R^{\theta}):=\{u\,:\, I_R^{\theta}\to \mathbb C \mathrm{\ of\ class\ } C^k;\ \sum_{0\leq j\leq k}\sup_{x\in I_R^{\theta}}\lvert u^{(j)}(x)\rvert\leq +\infty\},
\end{align*}
and $\lVert u\rVert_{C_b^0(I_L)}:=\sum_{0\leq j\leq k}\sup_{x\in I_L}\lvert u(x)^{(j)}\rvert$, $\lVert u\rVert_{C_b^0(I_L)}:=\sum_{0\leq j\leq k}\sup_{x\in I_R^{\theta}}\lvert u(x)^{(j)}\rvert$.
We also define the fundamental solutions $K_{j,L}: C_b^0(I_L)\to C^2_b(I_L)$ and $K_{j,R}: C_b^0(I_R^{\theta})\to C^2_b(I_R^{\theta})$ of $P_j-E$ on $I_L$ and $I_R^{\theta}$ as,
$$K_{j,L}[v](x):=\frac{1}{h^2\mathcal W[u_{j,b}^+,u_{j,b}^-]}\left(u_{j,b}^+(x)\int^x_{-\infty}u_{j,b}^-(t)v(t)dt+u_{j,b}^-(x)\int^b_xu_{j,b}^+(t)v(t)dt\right),$$
for $v\in C_b^0(I_L)$ and
$$K_{j,R}[v](x):=\frac{1}{h^2\mathcal W[u_{j,c}^-,u_{j,c}^+]}\left(u_{j,c}^-(x)\int^x_{c}u_{j,c}^+(t)v(t)dt+u_{j,c}^+(x)\int^{+\infty}_xu_{j,c}^-(t)v(t)dt\right),$$
for $v\in C_b^0(I_R^{\theta})$ respectively.
Then one has the following lemma.
\begin{lem}\label{fundamental3}
As $h$ goes to $0_+$ one has,
\begin{equation}\label{myeq2.12}
\lVert hK_{2,L}W^*\rVert_{\mathcal L(C^0_b(I_L))}=\mathcal O(h),
\end{equation}
\begin{equation}\label{myeq2.13}
\lVert hK_{1,L}W\rVert_{\mathcal L(C^0_b(I_L))}=\mathcal O(h^{-1/6}),
\end{equation}
Moreover, there exist complex numbers $\eta_b,\theta_b,\theta'_b$ such that,
\begin{equation}\label{myeq2.12.1}
\begin{pmatrix}
hK_{2,L}W^*u(b)\\ h\partial_x (K_{2,L}W^*u)(b)
\end{pmatrix}=\eta_b
\begin{pmatrix}
u_{2,b}^+(b)\\ \partial_x u_{2,b}^+(b)
\end{pmatrix},
\end{equation}
\begin{equation}\label{myeq2.13.1}
\begin{pmatrix}
hK_{1,L}Wu(b)\\ h\partial_x (K_{1,L}Wu)(b)
\end{pmatrix}=\theta_b
\begin{pmatrix}
u_{1,b}^+(b)\\ \partial_x u_{1,b}^+(b)
\end{pmatrix},
\end{equation}
\begin{equation}\label{myeq2.15.0}
\begin{pmatrix}
hK_{1,L}Wu_{2,b}^-(b)\\ h\partial_x (K_{1,L}Wu_{2,b}^-)(b)
\end{pmatrix}=\theta'_b
\begin{pmatrix}
u_{1,b}^+(b)\\ \partial_x u_{1,b}^+(b)
\end{pmatrix},
\end{equation}
and
$$\lvert \eta_b\rvert\leq Ch^{5/6}\lVert u\rVert_{C^0_b(I_L)},\ \lvert \theta_b\rvert\leq Ch^{-1/6}\lVert u\rVert_{C^0_b(I_L)},\ \lvert \theta'_b\rvert\leq Ch^{5/6}.$$
\end{lem}
\begin{proof}
We set,
\begin{align*}
&\tilde u_{1,b}^{\pm}:=\max\{\lvert u_{1,b}^{\pm}\rvert, \lvert h\partial_x u_{1,b}^{\pm}\rvert\},\\
&\tilde u_{2,b}^{\pm}:=\max\{\lvert u_{2,b}^{\pm}\rvert, \lvert h\partial_x u_{2,b}^{\pm}\rvert\},\\
&\mathcal U_{1,L}(x,t):=\tilde u_{1,b}^+(x)\tilde u_{1,b}^-(t){\bf 1}_{\{t<x\}}+\tilde u_{1,b}^-(x)\tilde u_{1,b}^+(t){\bf 1}_{\{t>x\}},\\
&\mathcal U_{2,L}(x,t):=\tilde u_{2,b}^+(x)\tilde u_{2,b}^-(t){\bf 1}_{\{t<x\}}+\tilde u_{2,b}^-(x)\tilde u_{2,b}^+(t){\bf 1}_{\{t>x\}}.
\end{align*}
By an integration by parts we have,
$$
\lvert hK_{1,L}Wv(x)\rvert= \mathcal O(h^{-1/3})\Bigg( \int_{-\infty}^b\mathcal U_{1,L}(x,t)\lvert v(t)\rvert dt+h\mathcal U_{1,L}(x,b)\lvert v(b)\rvert \Bigg ),
$$
and
$$\lvert hK_{2,L}W^*v(x)\rvert= \mathcal O(h^{-1/3})\Bigg( \int_{-\infty}^b\mathcal U_{2,L}(x,t)\lvert v(t)\rvert dt+h\mathcal U_{2,L}(x,b)\lvert v(b)\rvert \Bigg ).$$

For any $x,t\leq b$ we have
$$\mathcal U_{2,L}(x,t)=\mathcal O(h^{1/3})e^{-\lvert \mathrm{Re}\int_t^x(V_2-E)^{1/2}\rvert/h}.$$
Hence we have $\mathcal U_{2,L}(x,b)=\mathcal O(h^{1/3})$ and there exists $\alpha>0$ such that,
$$\int_{-\infty}^b\mathcal U_{2,L}(x,t)dt=\mathcal O(h^{1/3})\int_{-\infty}^be^{-\alpha\lvert x-t\rvert/h}dt=\mathcal O(h^{4/3}),$$
which proves \eqref{myeq2.12}. By the similar calculation we obtain \eqref{myeq2.12.1}.

Next we estimate $\mathcal U_{1,L}$. For any $\delta>0$ small enough, there exists $\alpha>0$ such that,
\begin{itemize}
\item If $\lvert t-a\rvert\leq C_1h^{2/3}$ or $b-C_1h^{2/3}\leq t\leq b$, then for any $-\infty\leq x\leq b$,
$$\mathcal U_{1,L}(x,t)=\mathcal O(1).$$
\item If $a+C_1h^{2/3}\leq t\leq b-C_1h^{2/3}$, then for any $-\infty\leq x\leq b$
$$\mathcal U_{1,L}(x,t)=\mathcal O(h^{1/6})\lvert t-b\rvert^{-1/4}.$$
\item If $a-2\delta\leq t\leq a-C_1h^{2/3}$, then for any $-\infty\leq x\leq b$,
$$\mathcal U_{1,L}(x,t)=\mathcal O(h^{1/6})\lvert t-a\rvert^{-1/4}.$$
\item If $t\leq a-2\delta$ and $x\leq a-\delta$ then,
$$\mathcal U_{1,L}(x,t)=\mathcal O(h^{1/3})e^{-\alpha\lvert t-x\rvert /h}.$$
\item If $t\leq a-2\delta$ and $a-\delta\leq x\leq b $ then,
$$\mathcal U_{1,L}(x,t)=\mathcal O(h^{1/6})e^{-\alpha\lvert t-a+\delta\rvert /h}.$$
\end{itemize}
Hence we have $\mathcal U_{1,L}(x,b)=\mathcal O(1)$. Moreover, when $x\leq a-\delta$ we have
\begin{align*}
\int_{-\infty}^b\mathcal U_{1,L}(x,t)dt&=\mathcal O(h^{1/3})\int_{-\infty}^{a-2\delta}e^{-\alpha\lvert t-x\rvert /h}dt+\mathcal O(1)\int_{a-C_1h^{2/3}}^{a+C_1h^{2/3}}dt\\
&\quad +\mathcal O(1)\int_{b-C_1h^{2/3}}^{b}dt+\mathcal O(h^{1/6})\int_{a+C_1h^{2/3}}^{b-C_1h^{2/3}}\lvert t-b\rvert^{-1/4}dt\\
&\quad +\mathcal O(h^{1/6})\int_{a-2\delta}^{a-C_1h^{2/3}}\lvert t-a\rvert^{-1/4}dt=\mathcal O(h^{1/6}).
\end{align*}
When $a-\delta\leq x\leq b $ we have
\begin{align*}
\int_{-\infty}^b\mathcal U_{1,L}(x,t)dt&=\mathcal O(h^{1/6})\int_{-\infty}^{a-2\delta}e^{-\alpha\lvert t-a+\delta\rvert /h}dt+\mathcal O(1)\int_{a-C_1h^{2/3}}^{a+C_1h^{2/3}}dt\\
&\quad +\mathcal O(1)\int_{b-C_1h^{2/3}}^{b}dt+\mathcal O(h^{1/6})\int_{a+C_1h^{2/3}}^{b-C_1h^{2/3}}\lvert t-b\rvert^{-1/4}dt\\
&\quad +\mathcal O(h^{1/6})\int_{a-2\delta}^{a-C_1h^{2/3}}\lvert t-a\rvert^{-1/4}dt=\mathcal O(h^{1/6}),
\end{align*}
which completes the proof of \eqref{myeq2.13}. By the similar calculation we obtain \eqref{myeq2.13.1}.

Finally, we shall prove \eqref{myeq2.15.0}. By an integration by parts we have
$$
\lvert hK_{1,L}Wu_{2,b}^-(b)\rvert= \mathcal O(h^{-1/3})\Bigg( \int_{-\infty}^b\tilde u_{1,b}^-(t)\tilde u_{2,b}^-(t) dt+h\tilde u_{1,b}^-(b)\tilde u_{2,b}^-(b) \Bigg )\tilde u_{1,b}^+(b),
$$
Since $\tilde u_{2,b}^-(b)\leq Ch^{1/6}$ we have $\tilde u_{1,b}^-(b)\tilde u_{2,b}^-(b) \leq Ch^{1/6}$. As for the integral we have
$$\int_{-\infty}^b\tilde u_{1,b}^-(t)\tilde u_{2,b}^-(t) dt=\mathcal O(h^{1/6})\int_{-\infty}^be^{-\alpha\lvert t-b\rvert/h}dt=\mathcal O(h^{7/6}),$$
which proves the estimate for the first element of \eqref{myeq2.15.0}. The estimate for the derivative follows from the similar calculation.
\end{proof}

In the same way as for Lemma \ref{fundamental3} we have,
\begin{lem}\label{fundamental4}
As $h$ goes to $0_+$ one has,
\begin{equation*}
\lVert hK_{1,R}W\rVert_{\mathcal L(C^0_b(I_R^{\theta}))}=\mathcal O(h),
\end{equation*}
\begin{equation*}
\lVert hK_{2,R}W^*\rVert_{\mathcal L(C^0_b(I_R^{\theta}))}=\mathcal O(h^{-1/6}),
\end{equation*}
Moreover, there exist complex numbers $\eta_c,\theta_c,\theta'_c$ such that,
\begin{equation*}
\begin{pmatrix}
hK_{1,R}Wu(c)\\ h\partial_x (K_{1,R}Wu)(c)
\end{pmatrix}=\eta_c
\begin{pmatrix}
u_{1,c}^+(c)\\ \partial_x u_{1,c}^+(c)
\end{pmatrix},
\end{equation*}
\begin{equation*}
\begin{pmatrix}
hK_{2,R}W^*u(c)\\ h\partial_x (K_{2,R}W^*u)(c)
\end{pmatrix}=\theta_c
\begin{pmatrix}
u_{2,c}^+(c)\\ \partial_x u_{2,c}^+(c)
\end{pmatrix},
\end{equation*}
\begin{equation*}
\begin{pmatrix}
hK_{2,R}W^*u_{1,c}^-(c)\\ h\partial_x (K_{2,R}W^*u_{1,c}^-)(c)
\end{pmatrix}=\theta'_c
\begin{pmatrix}
u_{2,c}^+(c)\\ \partial_x u_{2,c}^+(c)
\end{pmatrix},
\end{equation*}
and
$$\lvert \eta_c\rvert\leq Ch^{5/6}\lVert u\rVert_{C^0_b(I_R^{\theta})},\ \lvert \theta_c\rvert\leq Ch^{-1/6}\lVert u\rVert_{C^0_b(I_R^{\theta})},\ \lvert \theta'_c\rvert\leq Ch^{5/6}.$$
\end{lem}

\section{Solutions to the system}\label{thirdsec}
In this section, we construct solutions to \eqref{myeq1.1} on each interval. By Lemma \ref{fundamental}, the operators $M_b:=h^2K_{2,b}W^*K_{1,b}W$, $M_b':=h^2K_{2,b}W^*K_{1,b}'W$ and $M_b'':=h^2K_{2,b}W^*K_{1,b}''W$ are $\mathcal O(h^{1/2})$ when acting on $C(v_{1,b}^{\pm},h)$, $C(v_{2,b}^+,h)$ and $C(v_{2,b}^-,h)$ respectively. Therefore, we can define
\begin{align*}
&w_{1,b}^{\pm}:=\begin{pmatrix}
v_{1,b}^{\pm}+hK_{1,b}W\sum_{j\geq 0}M_b^j(hK_{2,b}W^*v_{1,b}^{\pm})\\
-\sum_{j\geq 0}M_b^j(hK_{2,b}W^*v_{1,b}^{\pm})
\end{pmatrix},\\
&w_{2,b}^{+}:=\begin{pmatrix}
-hK_{1,b}'W\sum_{j\geq 0}(M_b')^jv_{2,b}^{+}\\
\sum_{j\geq 0}(M_b')^jv_{2,b}^{+}
\end{pmatrix},\\
&w_{2,b}^{-}:=\begin{pmatrix}
-hK_{1,b}''W\sum_{j\geq 0}(M_b'')^jv_{2,b}^{-}\\
\sum_{j\geq 0}(M_b'')^jv_{2,b}^{-}
\end{pmatrix}.
\end{align*}
These are solutions to \eqref{myeq1.1} on $I_b$ and we have,
$$
w_{1,b}^{\pm}\to \begin{pmatrix}
v_{1,b}^{\pm}\\
0
\end{pmatrix},\
w_{2,b}^{\pm}\to \begin{pmatrix}
0\\
v_{2,b}^{\pm}
\end{pmatrix},
$$
as $h\to 0_+$.

In the same way, the operators $M_c:=h^2K_{1,c}WK_{2,c}W^*$, $M_c':=h^2K_{1,c}WK_{2,c}'W^*$ and $M_c'':=h^2K_{1,c}WK_{2,c}''W^*$ are $\mathcal O(h^{1/2})$ when acting on $C(v_{2,c}^{\pm},h)$, $C(v_{1,c}^+,h)$ and $C(v_{1,c}^-,h)$ respectively. Therefore, we can define
\begin{align*}
&w_{1,c}^{+}:=\begin{pmatrix}
\sum_{j\geq 0}(M_c')^jv_{1,c}^{+}\\
-hK_{2,c}'W^*\sum_{j\geq 0}(M_c')^jv_{1,c}^{+}
\end{pmatrix},\\
&w_{1,c}^{-}:=\begin{pmatrix}
\sum_{j\geq 0}(M_c'')^jv_{1,c}^{-}\\
-hK_{2,c}''W\sum_{j\geq 0}(M_c'')^jv_{1,c}^{-}
\end{pmatrix},\\
&w_{2,c}^{\pm}:=\begin{pmatrix}
-\sum_{j\geq 0}M_c^j(hK_{1,c}Wv_{2,c}^{\pm})\\
v_{2,c}^{\pm}+hK_{2,c}W\sum_{j\geq 0}M_c^j(hK_{1,c}Wv_{2,c}^{\pm})
\end{pmatrix}.
\end{align*}
These are solutions to \eqref{myeq1.1} on $I_c$ and we have,
$$
w_{1,c}^{\pm}\to \begin{pmatrix}
v_{1,c}^{\pm}\\
0
\end{pmatrix},\
w_{2,c}^{\pm}\to \begin{pmatrix}
0\\
v_{2,c}^{\pm}
\end{pmatrix},
$$
as $h\to 0_+$.

The following lemmas are consequences of Lemma \ref{fundamental}, \ref{fundamental2} and the definitions of fundamental solutions. The estimates for the derivatives follow observing that if $v$ is any of the functions $v_{j,b}^{\pm},\ v_{j,c}^{\pm},\ j=1,2$, the order of $\partial v(0)$ is that of $v(0)$ times $h^{-1}$.
\begin{lem}\label{at0b}
There exist complex numbers $\gamma_{1,b}^{\pm}$, $\gamma_{2,b}^+$ and $\delta_{j,b}^{\pm},\ j=1,2$ such that,
\begin{align*}
&w_{1,b}^{\pm}(0)=\begin{pmatrix}
v_{1,b}^{\pm}(0)+\gamma_{1,b}^{\pm}v_{1,b}^-(0)\\
\delta_{1,b}^{\pm}v_{2,b}^-(0)
\end{pmatrix}=\begin{pmatrix}
v_{1,b}^{\pm}(0)\\
0
\end{pmatrix}+\mathcal O(h^{1/2}\lvert v_{1,b}^{\pm}(0)\rvert),\\
&\partial w_{1,b}^{\pm}(0)=\begin{pmatrix}
\partial v_{1,b}^{\pm}(0)+\gamma_{1,b}^{\pm}\partial v_{1,b}^-(0)\\
\delta_{1,b}^{\pm}\partial v_{2,b}^-(0)
\end{pmatrix}=\begin{pmatrix}
\partial v_{1,b}^{\pm}(0)\\
0
\end{pmatrix}+\mathcal O(h^{1/2}\lvert \partial v_{1,b}^{\pm}(0)\rvert),\\
&w_{2,b}^{+}(0)=\begin{pmatrix}
\gamma_{2,b}^+v_{1,b}^+(0)+\zeta_{2,b}^+v_{1,b}^-(0)\\
v_{2,b}^{+}(0)+\delta_{2,b}^+v_{2,b}^-(0)
\end{pmatrix}=\begin{pmatrix}
0\\
v_{2,b}^{+}(0)
\end{pmatrix}+\mathcal O(h^{1/2}\lvert v_{2,b}^{+}(0)\rvert),\\
&\partial w_{2,b}^{+}(0)=\begin{pmatrix}
\gamma_{2,b}^+\partial v_{1,b}^+(0)+\zeta_{2,b}^+\partial v_{1,b}^-(0)\\
\partial v_{2,b}^{+}(0)+\delta_{2,b}^+\partial v_{2,b}^-(0)
\end{pmatrix}=\begin{pmatrix}
0\\
\partial v_{2,b}^{+}(0)
\end{pmatrix}+\mathcal O(h^{1/2}\lvert \partial v_{2,b}^{+}(0)\rvert),\\
&w_{2,b}^{-}(0)=\begin{pmatrix}
0\\
v_{2,b}^{-}(0)+\delta_{2,b}^-v_{2,b}^-(0)
\end{pmatrix}=\begin{pmatrix}
0\\
v_{2,b}^{-}(0)
\end{pmatrix}+\mathcal O(h^{1/2}\lvert v_{2,b}^{-}(0)\rvert).\\
&\partial w_{2,b}^{-}(0)=\begin{pmatrix}
0\\
\partial v_{2,b}^{-}(0)+\delta_{2,b}^-\partial v_{2,b}^-(0)
\end{pmatrix}=\begin{pmatrix}
0\\
\partial v_{2,b}^{-}(0)
\end{pmatrix}+\mathcal O(h^{1/2}\lvert \partial v_{2,b}^{-}(0)\rvert).\\
\end{align*}
\end{lem}

\begin{lem}\label{at0c}
There exist complex numbers $\gamma_{1,b}^{\pm}$, $\gamma_{2,b}^+$ and $\delta_{j,b}^{\pm},\ j=1,2$ such that,
\begin{align*}
&w_{2,c}^{\pm}(0)=\begin{pmatrix}
\delta_{2,c}^{\pm}v_{1,c}^-(0)\\
v_{2,c}^{\pm}(0)+\gamma_{2,c}^{\pm}v_{2,c}^-(0)
\end{pmatrix}=\begin{pmatrix}
0\\
v_{2,c}^{\pm}(0)
\end{pmatrix}+\mathcal O(h^{1/2}\lvert v_{2,c}^{\pm}(0)\rvert),\\
&\partial w_{2,c}^{\pm}(0)=\begin{pmatrix}
\delta_{2,c}^{\pm}\partial v_{1,c}^-(0)\\
\partial v_{2,c}^{\pm}(0)+\gamma_{2,c}^{\pm}\partial v_{2,c}^-(0)
\end{pmatrix}=\begin{pmatrix}
0\\
\partial v_{2,c}^{\pm}(0)
\end{pmatrix}+\mathcal O(h^{1/2}\lvert \partial v_{2,c}^{\pm}(0)\rvert),\\
&w_{1,c}^{+}(0)=\begin{pmatrix}
v_{1,c}^{+}(0)+\delta_{1,c}^+v_{1,c}^-(0)\\
\gamma_{1,c}^+v_{2,c}^+(0)+\zeta_{1,c}^+v_{2,c}^-(0)
\end{pmatrix}=\begin{pmatrix}
v_{1,c}^{+}(0)\\
0
\end{pmatrix}+\mathcal O(h^{1/2}\lvert v_{1,c}^{+}(0)\rvert),\\
&\partial w_{1,c}^{+}(0)=\begin{pmatrix}
\partial v_{1,c}^{+}(0)+\delta_{1,c}^+\partial v_{1,c}^-(0)\\
\gamma_{1,c}^+\partial v_{2,c}^+(0)+\zeta_{1,c}^+\partial v_{2,c}^-(0)
\end{pmatrix}=\begin{pmatrix}
\partial v_{1,c}^{+}(0)\\
0
\end{pmatrix}+\mathcal O(h^{1/2}\lvert \partial v_{1,c}^{+}(0)\rvert),\\
&w_{1,c}^{-}(0)=\begin{pmatrix}
v_{1,c}^{-}(0)+\delta_{1,c}^-v_{1,c}^-(0)\\
0
\end{pmatrix}=\begin{pmatrix}
v_{1,c}^{-}(0)\\
0
\end{pmatrix}+\mathcal O(h^{1/2}\lvert v_{1,c}^{-}(0)\rvert).\\
&\partial w_{1,c}^{-}(0)=\begin{pmatrix}
\partial v_{1,c}^{-}(0)+\delta_{1,c}^-\partial v_{1,c}^-(0)\\
0
\end{pmatrix}=\begin{pmatrix}
\partial v_{1,c}^{-}(0)\\
0
\end{pmatrix}+\mathcal O(h^{1/2}\lvert \partial v_{1,c}^{-}(0)\rvert).\\
\end{align*}
\end{lem}

We also define the solutions on $I_L$ and $I_R^{\theta}$. By Lemma \ref{fundamental3} the operator $M_L:=h^2K_{1,L}WK_{2,L}W^*$ and $M_R:=h^2K_{2,R}W^*K_{1,R}W$ are $\mathcal O(h^{2/3})$ when acting on $C^0_b(I_L)$ and $C^0_b(I_R^{\theta})$ respectively. Thus we can define,
\begin{align*}
&w_{1,L}:=\begin{pmatrix}
\sum_{j\geq 0}M_L^ju_{1,b}^{-}\\
-hK_{2,L}W^*\sum_{j\geq 0}M_L^ju_{1,b}^{-}
\end{pmatrix},\\
&w_{2,L}:=\begin{pmatrix}
-\sum_{j\geq 0}M_L^j(hK_{1,L}Wu_{2,b}^{-})\\
u_{2,b}^{-}+hK_{2,L}W^*\sum_{j\geq 0}M_L^j(hK_{1,L}Wu_{2,b}^{-})
\end{pmatrix},\\
\end{align*}
on $I_L$ and
\begin{align*}
&w_{1,R}:=\begin{pmatrix}
u_{1,c}^{-}+hK_{1,R}W\sum_{j\geq 0}M_R^j(hK_{2,R}W^*u_{1,c}^{-})\\
-\sum_{j\geq 0}M_R^j(hK_{2,R}W^*u_{1,c}^{-})
\end{pmatrix},\\
&w_{2,R}:=\begin{pmatrix}
-hK_{1,R}W\sum_{j\geq 0}M_R^ju_{2,c}^{-}\\
\sum_{j\geq 0}M_R^ju_{2,c}^{-}
\end{pmatrix},
\end{align*}
on $I_R^{\theta}$.

For these solutions we have the following proposition corresponding to \cite[Proposition 4.1]{FMW}.
\begin{pro}
The solutions $w_{j,L}$ and $w_{j,R}$, $(j=1,2)$ satisfy,
$$w_{j,L}\in L^2(I_L)\oplus L^2(I_L)\ ;\ w_{j,R}\in L^2(I_R^{\theta})\oplus L^2(I_R^{\theta}).$$
\end{pro}

\section{Connection of the solutions}\label{fourthsec}
In this section we investigate the connection of the basis $w_{j,b}^{\pm}$ and $w_{j,c}^{\pm}$ and that of $w_{j,L}$ (resp., $w_{j,R})$ and $w_{j,b}^{\pm}$ (resp., $w_{j,c}^{\pm}$). We define the $4\times4$ transition matrix $T$ as follows:
$$\begin{pmatrix}
w_{1,b}^+\\
w_{1,b}^-\\
w_{2,b}^+\\
w_{2,b}^-
\end{pmatrix}
=T\begin{pmatrix}
w_{1,c}^+\\
w_{1,c}^-\\
w_{2,c}^+\\
w_{2,c}^-
\end{pmatrix}.$$

From now on, we set,
\begin{align*}
&A_1:=\int_b^0\sqrt{V_1(t)-E}dt/h,\ B_1=\int_0^c\sqrt{V_1(t)-E}dt/h,\\
&A_2:=\int_0^c\sqrt{V_2(t)-E}dt/h,\ B_2=\int_b^0\sqrt{V_2(t)-E}dt/h.
\end{align*}
\begin{lem}\label{connection}
One has
\begin{align*}
T&=\begin{pmatrix}
t_{11}&t_{12}&t_{13}&t_{14}\\
t_{21}&t_{22}&t_{23}&t_{24}\\
t_{31}&t_{32}&t_{33}&t_{34}\\
t_{41}&t_{42}&t_{43}&t_{44}
\end{pmatrix}\\
&=\begin{pmatrix}
\mathcal O(h^{1/2}e^{A_1-B_1})&\mathcal O(e^{A_1+B_1})&\mathcal O(h^{1/2}e^{A_1-A_2})&\mathcal O(he^{A_1+A_2})\\
\mathcal O(e^{-A_1-B_1})&\mathcal O(h^{1/2}e^{-A_1+B_1})&\mathcal O(h^{1/2}e^{-A_1-A_2})&\mathcal O(h^{1/2}e^{-A_1+A_2})\\
\mathcal O(h^{1/2}e^{B_2-B_1})&\mathcal O(h^{1/2}e^{B_2+B_1})&\mathcal O(h^{1/2}e^{B_2-A_2})&\mathcal O(e^{B_2+A_2})\\
0&\mathcal O(h^{1/2}e^{-B_2+B_1})&\mathcal O(e^{-B_2-A_2})&\mathcal O(h^{1/2}e^{-B_2+A_2})
\end{pmatrix}.
\end{align*}
Moreover, one has,
\begin{align}
&t_{12}=(1+\mathcal O(h^{1/2}))e^{A_1+B_1},\label{myeq4.0.1}\\
&t_{34}=(1+\mathcal O(h^{1/2}))e^{A_2+B_2},\label{myeq4.0.2}\\
&t_{23}=-h^{1/2}\sqrt{\pi}e^{-A_1-A_2}(V_1(0)-E)^{-1/4}(V_1'(0)-V_2'(0))^{-1/2}\label{myeq4.0.3}\\
&\qquad\cdot(r_0(0)+r_1(0)\sqrt{V_1(0)-E})+\mathcal O(he^{-A_1-A_2}),\notag\\
&t_{32}=h^{1/2}\sqrt{\pi}e^{B_1+B_2}(V_1(0)-E)^{-1/4}(V_1'(0)-V_2'(0))^{-1/2}\label{myeq4.0.4}\\
&\qquad\cdot(r_0(0)+r_1(0)\sqrt{V_1(0)-E})+\mathcal O(he^{B_1+B_2}).\notag
\end{align}
\end{lem}
\begin{proof}
First, we consider $t_{12}$. By Cramer's formula we see that
\begin{equation}\label{myeq4.0.9}
t_{12}=\frac{\mathcal W[w_{1,c}^+,w_{1,b}^+,w_{2,c}^+,w_{2,c}^-]}{\mathcal W[w_{1,c}^+,w_{1,c}^-,w_{2,c}^+,w_{2,c}^-]}.
\end{equation}
For the calculation of Wronskians, we will use the following notation:
If $w$ is any of the vectors of functions $w_{j,b}^{\pm},w_{j,c}^{\pm},\ j=1,2$ and written as
$$w(x)=\begin{pmatrix} w_1(x)\\ w_2(x)\end{pmatrix},$$
we set,
$$\mathbf{w}(x):=\begin{pmatrix}w_1(x)\\ \partial w_1(x)\\ w_2(x)\\ \partial w_2(x)\end{pmatrix}.$$

To estimate the remainder terms in $\mathbf{w}$ we notice
\begin{equation}\label{myeq4.1}
\begin{split}
&v_{1,b}^{\pm} (0)=\mathcal O(h^{1/6}e^{\pm A_1}),\ \partial v_{1,b}^{\pm}(0)=\mathcal O(h^{-5/6}e^{\pm A_1}),\\
&v_{2,b}^{\pm} (0)=\mathcal O(h^{1/6}e^{\pm B_2}),\ \partial v_{2,b}^{\pm}(0)=\mathcal O(h^{-5/6}e^{\pm B_2}),\\
&v_{1,c}^{\pm} (0)=\mathcal O(h^{1/6}e^{\pm B_1}),\ \partial v_{1,c}^{\pm}(0)=\mathcal O(h^{-5/6}e^{\pm B_1}),\\
&v_{2,c}^{\pm} (0)=\mathcal O(h^{1/6}e^{\pm A_2}),\ \partial v_{2,c}^{\pm}(0)=\mathcal O(h^{-5/6}e^{\pm A_2}).
\end{split}
\end{equation}
By Lemma \ref{at0b}, Lemma \ref{at0c} and \eqref{myeq4.1} we have
\begin{equation}\label{myeq4.2}
\begin{split}
&\mathbf{w}_{1,b}^{\pm}(0)=\begin{pmatrix}
v_{1,b}^{\pm}(0)\\ \partial v_{1,b}^{\pm}(0)\\ 0\\ 0
\end{pmatrix}+\begin{pmatrix}
\mathcal O(h^{2/3}e^{{\pm}A_1})\\ \mathcal O(h^{-1/3}e^{{\pm}A_1})\\ \mathcal O(h^{2/3}e^{{\pm}A_1})\\ \mathcal O(h^{-1/3}e^{{\pm}A_1})
\end{pmatrix},\\
&\mathbf{w}_{2,b}^{\pm}(0)=\begin{pmatrix}
0\\ 0\\ v_{2,b}^{\pm}(0)\\ \partial v_{2,b}^{\pm}(0)
\end{pmatrix}+\begin{pmatrix}
\mathcal O(h^{2/3}e^{{\pm}B_2})\\ \mathcal O(h^{-1/3}e^{{\pm}B_2})\\ \mathcal O(h^{2/3}e^{{\pm}B_2})\\ \mathcal O(h^{-1/3}e^{{\pm}B_2})
\end{pmatrix},\\
&\mathbf{w}_{1,c}^{\pm}(0)=\begin{pmatrix}
v_{1,c}^{\pm}(0)\\ \partial v_{1,c}^{\pm}(0)\\ 0\\ 0
\end{pmatrix}+\begin{pmatrix}
\mathcal O(h^{2/3}e^{\pm B_1})\\ \mathcal O(h^{-1/3}e^{\pm B_1})\\ \mathcal O(h^{2/3}e^{\pm B_1})\\ \mathcal O(h^{-1/3}e^{\pm B_1})
\end{pmatrix},\\
&\mathbf{w}_{2,c}^{\pm}(0)=\begin{pmatrix}
0\\ 0\\ v_{2,c}^{\pm}(0)\\ \partial v_{2,c}^{\pm}(0)
\end{pmatrix}+\begin{pmatrix}
\mathcal O(h^{2/3}e^{\pm A_2})\\ \mathcal O(h^{-1/3}e^{\pm A_2})\\ \mathcal O(h^{2/3}e^{\pm A_2})\\ \mathcal O(h^{-1/3}e^{\pm A_2})
\end{pmatrix}.
\end{split}
\end{equation}
The Wronskian $\mathcal W[w_{1,c}^+,w_{1,c}^-,w_{2,c}^+,w_{2,c}^-]$ is written as the determinant
$$\mathcal W[w_{1,c}^+,w_{1,c}^-,w_{2,c}^+,w_{2,c}^-]=\mathrm{det}(\mathbf{w}_{1,c}^+(0),\mathbf{w}_{1,c}^-(0),\mathbf{w}_{2,c}^+(0),\mathbf{w}_{2,c}^-(0)).$$

We calculate the determinant using the multilinearity with respect to columns. We regard each column in the determinant as the sum of vectors whose upper or lower two elements are $0$, and expand the determinant into determinants whose columns are such vectors. Then the order of the upper and lower two elements in the remainder terms in \eqref{myeq4.2} are those of the upper or lower two elements of the leading terms multiplied by $h^{1/2}$.
Thus by \eqref{myeq2.0.1} and the definitions of $v_{j,b}^{\pm}$ and $v_{j,c}^{\pm}$, we have
\begin{align}
\mathcal W[w_{1,c}^+,w_{1,c}^-,w_{2,c}^+,w_{2,c}^-]&=\mathcal W[v_{1,c}^+,v_{1,c}^-]\mathcal W[v_{2,c}^+,v_{2,c}^-]+\mathcal O(h^{-5/6})\label{myeq4.1.1}\\
&=\frac{4}{\pi^2h^{4/3}}+\mathcal O(h^{-5/6}),\notag \\
\mathcal W[w_{1,c}^+,w_{1,b}^+,w_{2,c}^+,w_{2,c}^-]&=\mathcal W[v_{1,c}^+,v_{1,b}^+]\mathcal W[v_{2,c}^+,v_{2,c}^-]+\mathcal O(h^{-5/6}e^{A_1+B_1})\label{myeq4.1.2}\\
&=\frac{4}{\pi^2h^{4/3}}e^{A_1+B_1}+\mathcal O(h^{-5/6}e^{A_1+B_1}).\notag
\end{align}
By \eqref{myeq4.0.9}, \eqref{myeq4.1.1} and \eqref{myeq4.1.2} we obtain \eqref{myeq4.0.1}. By the similar calculation we obtain \eqref{myeq4.0.2}.

Next we study $t_{23}$ and $t_{32}$. For the calculation of higher order terms we notice
\begin{equation}
\begin{split}
&\mathbf{w}_{1,b}^-(0)=\begin{pmatrix}
v_{1,b}^-(0)\\ \partial v_{1,b}^-(0)\\ -hK_{2,b}W^*v_{1,b}^-(0)\\-h\partial (K_{2,b}W^*v_{1,b}^-)(0)
\end{pmatrix}+\begin{pmatrix}
\mathcal O(h^{2/3}e^{-A_1})\\ \mathcal O(h^{-1/3}e^{-A_1})\\ \mathcal O(h^{7/6}e^{-A_1})\\ \mathcal O(h^{1/6}e^{-A_1})
\end{pmatrix},\\
&\mathbf{w}_{2,b}^+(0)=\begin{pmatrix}
 -hK_{1,b}'Wv_{2,b}^+(0)\\-h\partial (K_{1,b}'Wv_{2,b}^+)(0)\\ v_{2,b}^+(0)\\ \partial v_{2,b}^+(0)
\end{pmatrix}+\begin{pmatrix}
\mathcal O(h^{7/6}e^{B_2})\\ \mathcal O(h^{1/6}e^{B_2})\\ \mathcal O(h^{2/3}e^{B_2})\\ \mathcal O(h^{-1/3}e^{B_2})\\ 
\end{pmatrix},\\
&\mathbf{w}_{1,c}^+(0)=\begin{pmatrix}
v_{1,c}^+(0)\\ \partial v_{1,c}^+(0)\\ -hK_{2,c}'W^*v_{1,c}^+(0)\\-h\partial (K_{2,c}'W^*v_{1,c}^+)(0)
\end{pmatrix}+\begin{pmatrix}
\mathcal O(h^{2/3}e^{B_1})\\ \mathcal O(h^{-1/3}e^{B_1})\\ \mathcal O(h^{7/6}e^{B_1})\\ \mathcal O(h^{1/6}e^{B_1})
\end{pmatrix},\\
&\mathbf{w}_{2,c}^-(0)=\begin{pmatrix}
 -hK_{1,c}Wv_{2,c}^-(0)\\-h\partial (K_{1,c}Wv_{2,c}^-)(0)\\ v_{2,c}^-(0)\\ \partial v_{2,c}^-(0)
\end{pmatrix}+\begin{pmatrix}
\mathcal O(h^{7/6}e^{-A_2})\\ \mathcal O(h^{1/6}e^{-A_2})\\ \mathcal O(h^{2/3}e^{-A_2})\\ \mathcal O(h^{-1/3}e^{-A_2})\\ 
\end{pmatrix}.\\
\end{split}
\end{equation}
Since $\mathcal W[v_{1,c}^+,v_{1,b}^-]=0$, we have
\begin{equation}\label{myeq4.3}
\begin{split}
\mathcal W[w_{1,c}^+,w_{1,c}^-,w_{1,b}^-,w_{2,c}^-]&=\mathcal W[v_{1,c}^+,v_{1,c}^-]\mathcal W[-hK_{2,b}W^*v_{1,b}^-,v_{2,c}^-]\\
&\quad+\mathcal W[v_{1,c}^-,v_{1,b}^-]\mathcal W[-hK_{2,c}'W^*v_{1,c}^+,v_{2,c}^-]\\
&\quad+\mathcal O(h^{-1/3}e^{-A_1-A_2})\\
&=-h^{-1}\mathcal W[v_{1,c}^+,v_{1,c}^-]e^{-A_2-B_2}\int_b^0v_{2,b}^-(t)W^*v_{1,b}^-(t)dt\\
&\quad +h^{-1}\mathcal W[v_{1,c}^-,v_{1,b}^-]\int_0^cv_{2,c}^-(t)W^*v_{1,c}^+(t)dt\\
&\quad +\mathcal O(h^{-1/3}e^{-A_1-A_2})\\
&=h^{-1}\mathcal W[v_{1,c}^-,v_{1,c}^+]\int_b^cu_{2,L}^-(t)W^*u_{1,R}^-(t)dt\\
&\quad+\mathcal O(h^{-1/3}e^{-A_1-A_2}).
\end{split}
\end{equation}
As for the derivative of $u_{1,R}^-$, for $x$ near $0$ we have
$$\partial u_{1,R}^-(x)=(1+\mathcal O(h))\frac{1}{h^{5/6}\sqrt{\pi}}(V_1(x)-E)^{3/4}e^{-\int_b^x\sqrt{V_1(t)-E}dt/h}.$$
We apply the stationary phase theorem to $\int_b^cu_{2,L}^-(t)W^*u_{1,R}^-(t)dt$. Estimating the derivatives of  $u_{2,L}^-(t)W^*u_{1,R}^-(t)e^{\int_b^t\sqrt{V_1(r)-E}dr/h+\int_t^c\sqrt{V_2(r)-E}dr/h}$ near $0$ by Cauchy's integral formula, we have
\begin{equation}\label{myeq4.4}
\begin{split}
\int_b^cu_{2,L}^-(t)W^*u_{1,R}^-(t)dt&=\frac{2h^{5/6}}{\sqrt{\pi}}(V_1(0)-E)^{-1/4}(V_1'(0)-V_2'(0))^{-1/2}\\
&\quad \cdot(r_0(0)+r_1(0)\sqrt{V_1(0)-E})e^{-A_1-A_2}\\
&\quad+\mathcal O(h^{4/3}e^{-A_1-A_2}).
\end{split}
\end{equation}
Here we used $V_1(0)=V_2(0)$. From \eqref{myeq4.3}, \eqref{myeq4.4} and that
$$t_{23}=\frac{\mathcal W[w_{1,c}^+,w_{1,c}^-,w_{1,b}^-,w_{2,c}^-]}{\mathcal W[w_{1,c}^+,w_{1,c}^-,w_{2,c}^+,w_{2,c}^-]},$$
\eqref{myeq4.0.3} follows. In the similar way we can obtain \eqref{myeq4.0.4}.

As for $t_{41}$, since by Lemmas \ref{at0b} and \ref{at0c} upper two elements of $\mathbf{w}_{2,b}^-$, $\mathbf{w}_{1,c}^-$ and $\mathbf{w}_{2,c}^{\pm}$ can be written as
$$C\begin{pmatrix}
v_{1,c}^-\\ \partial v_{1,c}^-
\end{pmatrix},$$
we have $t_{41}=0$. As for $t_{14}$, since we can see by Lemmas \ref{at0b} and \ref{at0c} that we need to choose remainder terms from at least two columns of the determinant $\mathrm{det}(\mathbf{w}_{1,c}^+,\mathbf{w}_{1,c}^-,\mathbf{w}_{2,c}^+,\mathbf{w}_{1,b}^+)$, we obtain $t_{14}=\mathcal O(he^{A_1+A_2})$. The estimates for the other terms can be obtained by the similar way as above.
\end{proof}

To make some of the elements of the transition matrix $0$, we change the basis on $I_b$ and $I_c$.
\begin{lem}\label{connection2}
There exist complex numbers $a_j,b_j,c_j,d_j,\ j=1,2$ such that if we set,
\begin{align*}
&\tilde w_{1,b}^+=w_{1,b}^+,\ \tilde w_{2,b}^+=w_{2,b}^+,\\
&\tilde w_{1,b}^-=w_{1,b}^--\tilde a_1w_{1,b}^+-\tilde a_2w_{2,b}^+,\\
&\tilde w_{2,b}^-=w_{2,b}^--\tilde b_1w_{1,b}^+-\tilde b_2w_{2,b}^+,\\
&\tilde w_{1,c}^+=w_{1,c}^+,\ \tilde w_{2,c}^+=w_{2,c}^+,\\
&\tilde w_{1,c}^-=w_{1,c}^-+\tilde c_1w_{1,c}^++\tilde c_2w_{2,c}^+,\\
&\tilde w_{2,c}^-=w_{2,c}^-+\tilde d_1w_{1,c}^++\tilde d_2w_{2,c}^+,
\end{align*}
the matrix $\tilde T$ defined by
$$\begin{pmatrix}
\tilde w_{1,b}^+\\
\tilde w_{1,b}^-\\
\tilde w_{2,b}^+\\
\tilde w_{2,b}^-
\end{pmatrix}
=\tilde T\begin{pmatrix}
\tilde w_{1,c}^+\\
\tilde w_{1,c}^-\\
\tilde w_{2,c}^+\\
\tilde w_{2,c}^-
\end{pmatrix},$$
has the following form;
$$\tilde T=\begin{pmatrix}
0&\tilde t_{12}&0&\tilde t_{14}\\
\tilde t_{21}&0&\tilde t_{23}&0\\
0&\tilde t_{32}&0&\tilde t_{34}\\
\tilde t_{41}&0&\tilde t_{43}&0\\
\end{pmatrix}.$$
with $\tilde t_{23}$ and $\tilde t_{32}$ having the same asymptotics as $t_{23}$ and $t_{32}$ respectively. Moreover, we have the following estimates.
\begin{align*}
&\tilde a_1=\mathcal O(h^{1/2}e^{-2A_1}),\ \tilde a_2=\mathcal O(h^{1/2}e^{-A_1-B_2}),\\
&\tilde b_1=\mathcal O(h^{1/2}e^{-A_1-B_2}),\ \tilde b_2=\mathcal O(h^{1/2}e^{-2B_2}),\\
&\tilde c_1=\mathcal O(h^{1/2}e^{-2B_1}),\ \tilde c_2=\mathcal O(h^{1/2}e^{-B_1-A_2}),\\
&\tilde d_1=\mathcal O(h^{1/2}e^{-B_1-A_2}),\ \tilde d_2=\mathcal O(h^{1/2}e^{-2A_2}).
\end{align*}
\end{lem}
\begin{proof}
We define $\tilde a_j$, $\tilde b_j$, $\tilde c_j$ and $\tilde d_j$ by
$$\begin{pmatrix}
\tilde a_1&\tilde a_2\\
\tilde b_1&\tilde b_2
\end{pmatrix}
\begin{pmatrix}
t_{12}&t_{14}\\
t_{32}&t_{34}
\end{pmatrix}=
\begin{pmatrix}
t_{22}&t_{24}\\
t_{42}&t_{44}
\end{pmatrix}.$$
$$\begin{pmatrix}
t_{12}&t_{14}\\
t_{32}&t_{34}
\end{pmatrix}
\begin{pmatrix}
\tilde c_1&\tilde c_2\\
\tilde d_1&\tilde d_2
\end{pmatrix}=
\begin{pmatrix}
t_{11}&t_{13}\\
t_{31}&t_{33}
\end{pmatrix}.$$
Then, it is easy to see that $\tilde T$ has the form as in the lemma, and by Lemma \ref{connection} we have
\begin{align*}
\begin{pmatrix}
\tilde a_1&\tilde a_2\\
\tilde b_1&\tilde b_2
\end{pmatrix}&=(t_{12}t_{34}-t_{14}t_{32})^{-1}
\begin{pmatrix}
t_{22}&t_{24}\\
t_{42}&t_{44}
\end{pmatrix}
\begin{pmatrix}
t_{34}&-t_{14}\\
-t_{32}&t_{12}
\end{pmatrix}\\
&=\begin{pmatrix}
\mathcal O(h^{1/2}e^{-2A_1})&\mathcal O(h^{1/2}e^{-A_1-B_2})\\
\mathcal O(h^{1/2}e^{-A_1-B_2})&\mathcal O(h^{1/2}e^{-2B_2})
\end{pmatrix},
\end{align*}
\begin{align*}
\begin{pmatrix}
\tilde c_1&\tilde c_2\\
\tilde d_1&\tilde d_2
\end{pmatrix}&=(t_{12}t_{34}-t_{14}t_{32})^{-1}
\begin{pmatrix}
t_{34}&-t_{14}\\
-t_{32}&t_{12}
\end{pmatrix}
\begin{pmatrix}
t_{11}&t_{13}\\
t_{31}&t_{33}
\end{pmatrix}\\
&=\begin{pmatrix}
\mathcal O(h^{1/2}e^{-2B_1})&\mathcal O(h^{1/2}e^{-B_1-A_2})\\
\mathcal O(h^{1/2}e^{-B_1-A_2})&\mathcal O(h^{1/2}e^{-2A_2})
\end{pmatrix}.
\end{align*}
\end{proof}

Note that $\tilde w_{j,b}^{\pm}$ and $\tilde w_{j,c}^{\pm}$ have the same asymptotics as $w_{j,b}^{\pm}$ and $w_{j,c}^{\pm}$ with respect to $h$ respectively. We next consider the connection of solutions at $b$ and $c$.

\begin{lem}\label{connection3}
Set $a_j^{\pm},b_j^{\pm},c_j^{\pm},d_j^{\pm}$ as follows:
\begin{equation}\label{myeq4.8}
\begin{split}
&w_{1,L}=D_L^{-1}(a_1^+\tilde w_{1,b}^++a_1^-\tilde w_{1,b}^-+a_2^+\tilde w_{2,b}^++a_2^-\tilde w_{2,b}^-),\\
&w_{2,L}=D_L^{-1}(b_1^+\tilde w_{1,b}^++b_1^-\tilde w_{1,b}^-+b_2^+\tilde w_{2,b}^++b_2^-\tilde w_{2,b}^-),\\
&w_{1,R}=D_R^{-1}(c_1^+\tilde w_{1,c}^++c_1^-\tilde w_{1,c}^-+c_2^+\tilde w_{2,c}^++c_2^-\tilde w_{2,c}^-),\\
&w_{2,R}=D_R^{-1}(d_1^+\tilde w_{1,c}^++d_1^-\tilde w_{1,c}^-+d_2^+\tilde w_{2,c}^++d_2^-\tilde w_{2,c}^-),
\end{split}
\end{equation}
where $D_L$ and $D_R$ are Wronskians 
\begin{align*}&D_L=\mathcal W[\tilde w_{1,b}^+,\tilde w_{1,b}^-,\tilde w_{2,b}^+,\tilde w_{2,b}^-],\\
&D_R=\mathcal W[\tilde w_{1,c}^+,\tilde w_{1,c}^-,\tilde w_{2,c}^+,\tilde w_{2,c}^-].
\end{align*}
Then we have
\begin{align}
&a_1^+=\frac{8}{\pi^2 h^{4/3}}\cos \frac{\mathcal A(E)}{h}+\mathcal O(h^{-5/6}),\label{myeq4.9}\\
&a_1^-=\frac{4}{\pi^2 h^{4/3}}\sin \frac{\mathcal A(E)}{h}+\mathcal O(h^{-5/6}),\label{myeq4.9.1}\\
&a_2^{\pm}=\mathcal O(h^{-1/2})\label{myeq4.9.2}\\
&b_1^{\pm}=\mathcal O(h^{-1/2})\label{myeq4.9.3}\\
&b_2^+=\frac{4}{\pi^2 h^{4/3}}+\mathcal O(h^{-5/6}),\label{myeq4.9.5}\\
&b_2^-=\mathcal O(h^{-1/2}),\label{myeq4.9.6}\\
&c_1^+=\frac{4}{\pi^2 h^{4/3}}+\mathcal O(h^{-5/6}),\label{myeq4.9.7}\\
&c_1^-=\mathcal O(h^{-1/2}),\label{myeq4.9.8}\\
&c_2^{\pm}=\mathcal O(h^{-1/2}),\label{myeq4.9.9}\\
&d_1^{\pm}=\mathcal O(h^{-1/2}),\label{myeq4.9.10}\\
&d_2^+=-\frac{4e^{i\frac{\pi}{4}}i}{\pi^2 h^{4/3}}+\mathcal O(h^{-5/6}),\label{myeq4.9.11}\\
&d_2^-=\frac{2e^{i\frac{\pi}{4}}}{\pi^2 h^{4/3}}+\mathcal O(h^{-5/6}).\label{myeq4.9.12}
\end{align}
\end{lem}
\begin{proof}
We start with $a_1^+$. By the Cramer's formula we have
$$a_1^+=\mathcal W[w_{1,L},\tilde w_{1,b}^-,\tilde w_{2,b}^+,\tilde w_{2,b}^-].$$
By Lemma \ref{connection2} there exist $\alpha>0$ such that
$$\mathcal W[w_{1,L},\tilde w_{1,b}^-,\tilde w_{2,b}^+,\tilde w_{2,b}^-]=\mathcal W[w_{1,L},w_{1,b}^-,w_{2,b}^+,w_{2,b}^-]+\mathcal O(e^{-\alpha/h}).$$
We use the notation $\mathbf{w}$ as in the proof of Lemma \ref{connection}. The Wronskian in the right-hand side is written as the determinant;
$$\mathcal W[w_{1,L},w_{1,b}^-,w_{2,b}^+,w_{2,b}^-]=\mathrm{det}(\mathbf{w}_{1,L}(b),\mathbf{w}_{1,b}^-(b),\mathbf{w}_{2,b}^+(b),\mathbf{w}_{2,b}^-(b)).$$
By Lemma \ref{fundamental3} we have
\begin{equation}\label{myeq4.10}
\mathbf{w}_{1,L}(b)=\begin{pmatrix}
u_{1,b}^-(b)\\
\partial u_{1,b}^-(b)\\
0\\
0
\end{pmatrix}+\begin{pmatrix}
\mathcal O(h^{2/3})\\
\mathcal O(1)\\
\mathcal O(h)\\
\mathcal O(1)
\end{pmatrix}
.
\end{equation}
By Lemma \ref{fundamental} we also have
\begin{equation}\label{myeq4.11}
\begin{split}
&\mathbf{w}_{1,b}^-(b)=\begin{pmatrix}
v_{1,b}^-(b)\\
\partial v_{1,b}^-(b)\\
0\\
0
\end{pmatrix}+\begin{pmatrix}
\mathcal O(h^{1/2})\\
\mathcal O(h^{-1/6})\\
\mathcal O(h)\\
\mathcal O(1)
\end{pmatrix},\\
&\mathbf{w}_{2,b}^+(b)=\begin{pmatrix}
0\\
0\\
v_{2,b}^+(b)\\
\partial v_{2,b}^+(b)
\end{pmatrix}+\begin{pmatrix}
0\\
0\\
\mathcal O(h^{2/3})\\
\mathcal O(h^{-1/3})
\end{pmatrix},\\
&\mathbf{w}_{2,b}^-(b)=\begin{pmatrix}
0\\
0\\
v_{2,b}^-(b)\\
\partial v_{2,b}^-(b)
\end{pmatrix}+\begin{pmatrix}
\mathcal O(h^{5/6})\\
\mathcal O(h^{1/6})\\
\mathcal O(h^{2/3})\\
\mathcal O(h^{-1/3})
\end{pmatrix}.
\end{split}
\end{equation}

We calculate the determinant using the multilinearity with respect to columns as in the proof of Lemma \ref{connection}. Then the order of the upper and lower two elements in the remainder terms in \eqref{myeq4.10} and \eqref{myeq4.11} are those of the leading terms multiplied by $h^{1/2}$. Thus we obtain
$$\mathcal W[w_{1,L},w_{1,b}^-,w_{2,b}^+,w_{2,b}^-]=\mathcal W[u_{1,b}^-,v_{1,b}^-]\mathcal W[v_{2,b}^+,v_{2,b}^-]+\mathcal O(h^{-5/6}).$$
From the definition we have $u_{1,b}^-=u_{1,L}^-$, $v_{1,b}^-=u_{1,R}^-$. Hence by Proposition \ref{LRC} we obtain
$$\mathcal W[u_{1,b}^-,v_{1,b}^-]=b_-\mathcal W[u_{1,R}^+,u_{1,R}^-]=-\frac{4}{\pi h^{2/3}}\cos\frac{\mathcal A(E)}{h}+\mathcal O(h^{1/3}).$$
Since $v_{2,b}^+=e^{S_2/h}u_{2,L}^-$ and $v_{2,b}^-=e^{-S_2/h}u_{2,L}^+$, we have
$$\mathcal W[v_{2,b}^+,v_{2,b}^-]=\mathcal W[u_{2,L}^-,u_{2,L}^+]=-\frac{2}{\pi h^{2/3}}(1+\mathcal O (h))$$
which completes the proof of \eqref{myeq4.9}. The proof of \eqref{myeq4.9.1} is similar. 

The estimates \eqref{myeq4.9.2} and \eqref{myeq4.9.3} follow from the calculation of the determinant as above, Lemma \ref{fundamental} and Lemma \ref{fundamental3}. We can prove \eqref{myeq4.9.5} by the similar calculation as in the proof of \eqref{myeq4.9}. The estimate \eqref{myeq4.9.6} is obtained using $\mathcal W[v_{2,b}^+,u_{2,b}^-]=0$. The estimates \eqref{myeq4.9.7}-\eqref{myeq4.9.12} are obtained in the same way as \eqref{myeq4.9}-\eqref{myeq4.9.6}.
\end{proof}

\section{Quantisation condition}\label{fifthsec}
\begin{pro}
$E\in \mathcal D_I$ is a resonance of $P$ if and only if,
\begin{equation}\label{myeq4.15}
\cos \frac{\mathcal A(E)}{h}+f(E,h)=hF(E,h).
\end{equation}
where $f(E,h)$ and $F(E,h)$ are analytic for $E\in \mathcal D_I$, $f(E,h)$ is real for real $E$ and
\begin{align*}
f(E,h)=&\mathcal O(h^{1/2}),\\
F(E,h)=&-\frac{\pi}{4i}\left(\sin \frac{\mathcal A(E)}{h}\right)e^{-2A_1-2A_2}(V_1(0)-E)^{-1/2}\\
&\cdot(V_1'(0)-V_2'(0))^{-1}(r_0(0)+r_1(0)\sqrt{V_1(0)-E})^2+\mathcal O(h^{1/2}e^{-2A_1-2A_2}).
\end{align*}
\end{pro}
\begin{proof}
As in \cite{FMW}, $E$ is a resonance if and only if $w_{1,L},w_{2,L},w_{1,R}$ and $w_{2,R}$ are linearly dependent, that is 
\begin{equation}\label{myeq4.11.1}
\mathcal W_0(E)=\mathcal W[w_{1,L},w_{2,L},w_{1,R},w_{2,R}]=0.
\end{equation}
We substitute the right-hand side of \eqref{myeq4.8} for $w_{j,L},w_{j,R}$ in \eqref{myeq4.11.1} and develop the Wronskian as a sum of terms of the form $C(h)\mathcal W[w_1,w_2,w_3,w_4]$ where $C(h)$ is a constant, $w_1,w_2$ and $w_3,w_4$ are chosen from $w_{j,b}^{\pm},\ (j=1,2)$ and $w_{j,c}^{\pm},\ (j=1,2)$ respectively. If only one of the $w_j,\ (j=1,\dots, 4)$ is chosen from $w_{j,b}^-$ or $w_{j,c}^-$, then by the form of $\tilde T$ in Lemma \ref{connection2} we can see that $\mathcal W[w_1,w_2,w_3,w_4]=0$.
Thus by Lemma \ref{connection3} we have
\begin{equation}\label{myeq4.12}
\begin{split}
D_L^{2}D_R^{2}\mathcal W_0(E)=&(a_1^+b_2^+-b_1^+a_2^+)(c_1^+d_2^+-d_1^+c_2^+)\mathcal W[\tilde w_{1,b}^+,\tilde w_{2,b}^+,\tilde w_{1,c}^+,\tilde w_{2,c}^+]\\
&+(a_1^-b_2^+-b_1^-a_2^+)(c_1^+d_2^--d_1^+c_2^-)\mathcal W[\tilde w_{1,b}^-,\tilde w_{2,b}^+,\tilde w_{1,c}^+,\tilde w_{2,c}^-]\\
&+\mathcal O(e^{\max\{A_1-B_1-A_2+B_2,\ -A_1+B_1+A_2-B_2\}}).
\end{split}
\end{equation}
By Lemma \ref{connection2} we can easily see that
\begin{align*}
\mathcal W[\tilde w_{1,b}^-,\tilde w_{2,b}^+,\tilde w_{1,c}^+,\tilde w_{2,c}^-]&=\tilde t_{23}\tilde t_{32}\mathcal W[\tilde w_{2,c}^+,\tilde w_{1,c}^-,\tilde w_{1,c}^+,\tilde w_{2,c}^-]\\
&=\tilde t_{23}\tilde t_{32}\mathcal W[w_{2,c}^+,w_{1,c}^-,w_{1,c}^+,w_{2,c}^-],
\end{align*}
and by Lemma \ref{connection3} we have
\begin{equation}\label{myeq4.13}
\begin{split}
&a_1^+b_2^+-b_1^+a_2^+=\frac{32}{\pi^4h^{8/3}}\left(\cos \frac{\mathcal A(E)}{h}+f(E,h)\right),\\
&c_1^+d_2^+-d_1^+c_2^+=-\frac{16e^{i\frac{\pi}{4}}i}{\pi^4h^{8/3}}(1+\mathcal O(h^{1/2})),\\
&a_1^-b_2^+-b_1^-a_2^+=\frac{16}{\pi^4h^{8/3}}\left(\sin \frac{\mathcal A(E)}{h}+\mathcal O(h^{1/2})\right),\\
&c_1^+d_2^--d_1^+c_2^-=\frac{8e^{i\frac{\pi}{4}}}{\pi^4h^{8/3}}(1+\mathcal O(h^{1/2})),\\
\end{split}
\end{equation}
where $\lvert f(E,h)\rvert=\mathcal O(h^{1/2})$ uniformly with respect to $E\in \mathcal D_I$. From the construction of $u_{j,L}^{\pm}$  and $u_{1,R}^{\pm}$ we can easily see that $f(E,h)$ is real for real $E$. 
We can also see by the similar calculation as in the proof of Lemma \ref{connection},
\begin{equation}\label{myeq4.14}
\begin{split}
&\mathcal W[\tilde w_{1,b}^+,\tilde w_{2,b}^+,\tilde w_{1,c}^+,\tilde w_{2,c}^+]=-\frac{4}{\pi^2h^{4/3}}e^{A_1+A_2+B_1+B_2}(1+\mathcal O(h^{1/2})),\\
&\mathcal W[w_{2,c}^+,w_{1,c}^-,w_{1,c}^+,w_{2,c}^-]=-\frac{4}{\pi^2h^{4/3}}(1+\mathcal O(h^{1/2})).
\end{split}
\end{equation}
By \eqref{myeq4.12}, \eqref{myeq4.13}, \eqref{myeq4.14} and asymptotics of $t_{23}$ and $t_{32}$, we have
\begin{align*}
D_L^{2}D_R^{2}\mathcal W_0(E)=&\frac{2^{11}e^{i\frac{\pi}{4}}i}{\pi^{10}h^{20/3}}\left(\cos \frac{\mathcal A(E)}{h}+f(E,h)\right)e^{A_1+A_2+B_1+B_2}(1+\mathcal O(h^{1/2}))\\
&+\frac{2^{9}e^{i\frac{\pi}{4}}}{\pi^{9}h^{17/3}}\left(\sin \frac{\mathcal A(E)}{h}\right)e^{-A_1-A_2+B_1+B_2}(V_1(0)-E)^{-1/2}\\
&\cdot(V_1'(0)-V_2'(0))^{-1}(r_0(0)+r_1(0)\sqrt{V_1(0)-E})^2\\
&+\mathcal O(h^{-31/6}e^{-A_1-A_2+B_1+B_2}).
\end{align*}
Therefore, $\mathcal W_0(E)=0$ is equivalent to
\begin{equation}
\begin{split}
\cos \frac{\mathcal A(E)}{h}+f(E,h)=&-\frac{h\pi}{4i}\left(\sin \frac{\mathcal A(E)}{h}\right)e^{-2A_1-2A_2}(V_1(0)-E)^{-1/2}\\
&\cdot(V_1'(0)-V_2'(0))^{-1}(r_0(0)+r_1(0)\sqrt{V_1(0)-E})^2\\
&+\mathcal O(h^{3/2}e^{-2A_1-2A_2}).
\end{split}
\end{equation}
\end{proof}

\section{Completion of the proof of Theorem \ref{main}}\label{sixthsec}
In order to solve \eqref{myeq4.15}, we first observe that the roots of $\cos (\mathcal A(E)/h)=0$ are given by $E=e_k(h)$ with,
$$e_k(h):=\mathcal A^{-1}\left((k+\frac{1}{2})\pi h\right)\in \mathbb R,\ k\in \mathbb Z.$$
For sufficiently small $d'>0$ we set $\tilde I=I+[-d',d']$. Then since there exist constants $m,M>0$ such that for $E\in \tilde I$ we have $m\leq\mathcal A'(E)\leq M$, the estimate $\lvert e_k(h)-e_l(h)\rvert\geq \pi h\lvert k-l\rvert/M$ holds.
Since $\mathcal A(E)$ is holomorphic in $\mathcal D_I$ (see, e.g., Fujii\'e-Ramond \cite{FR}), by the Cauchy's integral formula we have for fixed $C'>0$ and $z\in \mathbb C$ such that $\lvert z-e_k\rvert<hC'/2$ 
\begin{equation}\label{myeq4.15.1}
\begin{split}
\cos \frac{\mathcal A(z)}{h}=&-h^{-1}\mathcal A'(e_k)(z-e_k)\\
&+\frac{1}{2\pi i}\oint_{\lvert \zeta-e_k\rvert =C'h}\frac{1}{\zeta-z}\left (\frac{z-e_k}{\zeta-e_k}\right)^2\cos\frac{\mathcal A(\zeta)}{h}d\zeta.
\end{split}
\end{equation}
Since $\mathcal A'(e_k)>m$ for $e_k\in \tilde I$, $f(E,h)=\mathcal O(h^{1/2})$ and
$$\left \lvert\frac{1}{2\pi i}\oint_{\lvert \zeta-e_k\rvert =C'h}\frac{1}{\zeta-z}\left (\frac{z-e_k}{\zeta-e_k}\right)^2\cos\frac{\mathcal A(\zeta)}{h}d\zeta\right \rvert\leq Ch^{-2}\lvert z-e_k\rvert^2,$$
by the Rouch\'e's theorem we can see that for sufficiently large $C_0'>0$ and sufficiently small $h>0$, $\cos\mathcal A(z)/h+f(z,h)=0$ has a unique solution $\tilde e_k(h)$ in $B(e_k;C_0'h^{3/2})$ for $e_k\in \tilde I$ and conversely, all the roots in $\mathcal D_I$ are of this type. Since $f(E,h)$ is real for $E\in \mathbb R$ and by \eqref{myeq4.15.1} there exists a number $C>0$ such that $[-CC_0'h^{1/2},CC_0'h^{1/2}]\subset\{\cos (\mathcal A(E)/h);-C_0'h^{3/2}<E-e_k<C_0'h^{3/2}\}$, we can see that $\tilde e_k(h)$ is real.

Estimating the remainder terms by the Cauchy's integral formula as above, the left-hand side of \eqref{myeq4.15} is written as
\begin{equation}\label{myeq4.16}
\begin{split}
\cos \frac{\mathcal A(z)}{h}+f(z,h)=&-h^{-1}(\mathcal A'(\tilde e_k(h))+\mathcal O(h^{1/2}))\sin\frac{\mathcal A(\tilde e_k(h))}{h}(z-\tilde e_k(h))\\
&+\mathcal O(h^{-2}\lvert z-\tilde e_k(h)\rvert^2).
\end{split}
\end{equation}
Thus again by the Rouch\'e's theorem we can see that for sufficiently large $C_0''>0$ and sufficiently small $h>0$, $\mathcal W_0(E)=0$ has a unique solution $E_k(h)$ in
$$B(\tilde e_k;C_0''h^{2}e^{-2A_1-2A_2}),$$
and all the roots in $\mathcal D_I$ are of this type.
In the same way as above we also have
$$\sin \frac{\mathcal A(E_k(h))}{h}=\sin \frac{\mathcal A(\tilde e_k(h))}{h}+\mathcal O(he^{-2A_1-2A_2}).$$
Hence substituting $E_k(h)$ into \eqref{myeq4.15} for $z$ we can see
\begin{align*}
E_k(h)-\tilde e_k(h)=&-\frac{h^2\pi i}{4}\mathcal A'(e_k(h))^{-1}e^{-2A_1-2A_2}(V_1(0)-E_k(h))^{-1/2}(V_1'(0)-V_2'(0))^{-1}\\
&\cdot(r_0(0)+r_1(0)\sqrt{V_1(0)-E_k(h)})^2+\mathcal O(h^{5/2})e^{-2A_1-2A_2},
\end{align*}
from which Theorem \ref{main} follows.

%\begin{ack}
%\end{ack}

\section*{Acknowledgement}
The author deeply thanks Professor A. Martinez for suggesting the topic treated in this paper and for his helpful discussions.

\end{document}